\documentclass{IEEEtran}
\usepackage{cite}
\usepackage{amsmath,amssymb,amsfonts}
\usepackage{amsthm}
\usepackage{graphicx}
\usepackage{textcomp}
\usepackage{array}
\usepackage[caption=false,font=normalsize,labelfont=sf,textfont=sf]{subfig}
\usepackage{textcomp}
\usepackage{stfloats}
\usepackage{url}
\usepackage{verbatim}
\usepackage{graphicx}
\usepackage{cite}
\usepackage{flushend}

\usepackage[linesnumbered,ruled]{algorithm2e}
\usepackage{algorithmicx}

\newtheorem{theorem}{Theorem}
\newtheorem{proposition}{Proposition}

\newtheorem{lemma}{Lemma}
\theoremstyle{definition}
\newtheorem{definition}{Definition}
\theoremstyle{remark}

\def\BibTeX{{\rm B\kern-.05em{\sc i\kern-.025em b}\kern-.08em
    T\kern-.1667em\lower.7ex\hbox{E}\kern-.125emX}}
\begin{document}

\title{Learning Predictive Safety Filter via\\Decomposition of Robust Invariant Set}

\author{Zeyang Li, Chuxiong Hu, Weiye Zhao, Changliu Liu
    \thanks{Zeyang Li and Chuxiong Hu are with the Department of Mechanical Engineering, Tsinghua University, Beijing 100084, China (email: li-zy21@mails.tsinghua.edu.cn; cxhu@tsinghua.edu.cn).}
    \thanks{Weiye Zhao and Changliu Liu are with the Robotics Institute, Carnegie Mellon University, Pittsburgh, PA 15217 USA (email: weiyezha@andrew.cmu.edu; cliu6@andrew.cmu.edu).}
    \thanks{Corresponding author: Chuxiong Hu and Changliu Liu.}}


\maketitle


\begin{abstract}

    Ensuring safety of nonlinear systems under model uncertainty and external disturbances is crucial, especially for real-world control tasks. Predictive methods such as robust model predictive control (RMPC) require solving nonconvex optimization problems online, which leads to high computational burden and poor scalability. Reinforcement learning (RL) works well with complex systems, but pays the price of losing rigorous safety guarantee.
    This paper presents a theoretical framework that bridges the advantages of both RMPC and RL to synthesize safety filters for nonlinear systems with state- and action-dependent uncertainty. We decompose the robust invariant set (RIS) into two parts: a target set that aligns with terminal region design of RMPC, and a reach-avoid set that accounts for the rest of RIS. We propose a policy iteration approach for robust reach-avoid problems and establish its monotone convergence. This method sets the stage for an adversarial actor-critic deep RL algorithm, which simultaneously synthesizes a reach-avoid policy network, a disturbance policy network, and a reach-avoid value network. The learned reach-avoid policy network is utilized to generate nominal trajectories for online verification, which filters potentially unsafe actions that may drive the system into unsafe regions when worst-case disturbances are applied. We formulate a second-order cone programming (SOCP) approach for online verification using system level synthesis, which optimizes for the worst-case reach-avoid value of any possible trajectories. The proposed safety filter requires much lower computational complexity than RMPC and still enjoys persistent robust safety guarantee. The effectiveness of our method is illustrated through a numerical example.

\end{abstract}


\section{Introduction}

\IEEEPARstart {S}{afety} remains a paramount concern for various intricate real-world control tasks \cite{ames2019control, zhao2022back, lin2023real, zhou2023intelligent}.
Safety filters, acting as the ultimate safeguard for cyber-physical systems, adjust and modify the potentially unsafe control inputs so that the system states remain within the constraint sets \cite{liu2014control, ames2016control}.
While there is extensive research on safety filter design for nonlinear systems \cite{ames2019control}, guaranteeing safety in the face of model uncertainty and external disturbances remains a significant and challenging open problem \cite{wei2022persistently, breeden2023robust, wabersich2023data}.

The key to maintaining safety under disturbances is to constrain the system states inside robust invariant sets (RIS). The maximal RIS is referred to as viability kernel, out of which the system will inevitably violate the safety constraints under worst-case disturbances, regardless of the control inputs \cite{blanchini2008set}.
Hamilton-Jacobi reachability analysis is a fundamental theoretic tool for solving RIS, represented by the zero-sublevel set of safety value functions \cite{mitchell2005time, bansal2017hamilton}. By solving the Hamilton-Jacobi-Isaacs (HJI) partial differential equation (PDE), one obtains the optimal safety value function (identifying the maximal RIS) and the optimal safety policy (the minimax actions for the Hamiltonian).
A straightforward approach for robust safe control is using the solution of HJI PDE to filter unsafe control inputs \cite{fisac2018general, chen2021fastrack}.
However, the HJI PDE-based approach suffers from the curse of dimensionality, since it requires gridding the state spaces, which are intractable for high-dimensional systems \cite{bansal2017hamilton}.

Beyond addressing the HJI PDE, two predominant methodologies emerge for ensuring safety under disturbances in a more tractable manner: value-based approach and policy-based approach. The value-based approach implements control correction through online convex optimization (e.g., quadratic programming), leveraging values and derivatives of robustified energy functions, such as robust control barrier function and robust safety index \cite{breeden2023robust, buch2021robust, emam2022safe, wei2022persistently, wei2023robust}. The policy-based approach (also referred to as predictive methods \cite{wabersich2023data}) evaluates if a state can be safely guided back to terminal regions with some failsafe policy under worst-case disturbances, such as robust model predictive control (RMPC) and its variants \cite{koller2018learning, kohler2020computationally, wabersich2018linear, georgiou2022computationally, wabersich2021predictive}. However, both strategies come with notable limitations. The value-based approach hinges on the existence of valid robustified energy functions, which can be challenging to design or validate, otherwise the online optimization procedure can easily become infeasible, compromising the safety guarantee. As a result, hand-crafted energy functions tend to be overly conservative to maintain feasibility, leading to performance degeneration \cite{liu2023safe}. The policy-based approach demands computing a failsafe policy, typically synthesized through RMPC. This is in general computationally prohibitive, as RMPC inherently involves solving nonconvex optimization problems.

On the other hand, learning-based approaches have demonstrated empirical efficacy for safe nonlinear control. A promising line of work is safe reinforcement learning (safe RL) \cite{garcia15a, zhao2023state, achiam2017constrained, chow2017risk}. Its objective is to learn a policy that not only maximizes the reward but also adheres to safety constraints. Some safe RL methods also incorporate synthesizing neural network representations of energy functions or safety value function (of HJI PDE), to facilitate the constrained policy optimization process \cite{ma2022joint, yu2022reachability, li2023robust}. Another line of work is utilizing supervised learning techniques to synthesize neural energy functions, subsequently deploying them for online safety filtering \cite{qin2021learning, liu2023safe, dawson2022safe}. The datasets are collected through sampling across both safe and unsafe regions, and the neural approximators are optimized with specific loss functions that encapsulate the intrinsic properties of energy functions. Learning-based approaches typically enjoy considerable scalability and empirically work well in certain examples, credited to the representational power of neural networks. Nevertheless, there are also significant limitations for learning-based methods. First, most learning-based methods cannot generalize over different model uncertainty and external disturbances, resulting in vulnerable policies that are sensitive to discrepancies between training and deploying environments. Addressing robustness in safe RL or energy function learning remains a topic of active research and demands further exploration \cite{dawson2022safe, li2023robust}. Second, the merit of enhanced scalability comes with the price of losing theoretical safety guarantee. The unsafe rates of learned control policies or safety filters are actually nonzero in many cases \cite{ma2022joint}, which can cause damage to the systems in real-world control tasks.

\textbf{Contribution}: In this work, we establish a theoretical framework for the synthesis of policy-based safety filters, applicable to systems with state- and action-dependent uncertainty. Our method is able to retain the scalability and non-conservativeness of learning-based methods, yet uphold the rigorous safety guarantee of predictive methods such as RMPC. The (maximal) RIS of a system is decomposed into two components: a target set that aligns with terminal region design of RMPC, and a reach-avoid set that encapsulates the rest of RIS. Inspired by Hamilton-Jacobi formulation for reach-avoid differential games, we propose a policy iteration approach for robust reach-avoid problems and establish its monotone convergence to the maximal reach-avoid sets. This method paves the way for an adversarial actor-critic deep RL algorithm, which simultaneously synthesizes a reach-avoid policy network, a disturbance policy network, and a reach-avoid value network. The learned reach-avoid policy network is utilized to generate nominal trajectories for online verification, which filters potentially unsafe actions that may drive the system into unsafe regions when worst-case disturbances are applied. Drawing inspiration from tube-based MPC, a second-order cone programming (SOCP) approach for online verification is proposed using system level synthesis (SLS), which computes the robust invariant tube with the smallest reach-avoid value, taking advantage of additional linear stabilizing controllers. The effectiveness of our method is illustrated through a numerical example in Section \ref{Numerical Example}.

\textbf{Notation}: For a vector $a\in\mathbb{R}^p$, we denote its $i$-th component as $a^i$. For a matrix $A\in\mathbb{R}^{p\times q}$, we denote its $i$-th row as $A^{i,:}$ and the $j$-th to $k$-th element in the $i$-th row as $A^{i,j:k}$. The aforementioned indexes may also be used as subscripts when the superscripts are occupied by other information. The block diagonal matrix consisting of $A_1,\dots,A_n$ is denoted as ${\rm blkdiag}(A_1,\cdots,A_n)$. The diagonal matrix consisting of $a_1,\cdots,a_n$ is denoted as ${\rm diag}(a_1,\cdots,a_n)$. $0_{p,q}$ denote the zero matrix with $p$ rows and $q$ columns. $0_{p}$ and $I_{p}$ denote the square zero matrix and identity matrix with size $p$, respectively. We may drop the subscript and use $I$ for identity matrices with appropriate dimensions.
Let $\left\|\cdot\right\|_{1}$, $\left\|\cdot\right\|_{2}$ and $\left\|\cdot\right\|_{\infty}$ denote the one-norm, two-norm, infinity-norm of a vector or a matrix, respectively. Let $\mathcal{B}_{\infty}^{p}$ denote the unit ball with infinity-norm, i.e., $\mathcal{B}_{\infty}^p=\left\{x \in \mathbb{R}^p \mid\|x\|_{\infty} \leq 1\right\}$. The Kronecker product between matrix $A$ and $B$ is denoted as $A\otimes B$.
A block lower-triangular matrix (we refer to the set of such matrices as $\mathcal{L}^{T,p\times q}$) is denoted as
\begin{equation}
    \mathbf{R}=\begin{pmatrix}
        \mathbf{R}^{1,1} & & & \\
        \mathbf{R}^{2,1} & \mathbf{R}^{2,2} & & \\
        \vdots & \ddots & \ddots & \\
        \mathbf{R}^{T, 1} & \cdots & \mathbf{R}^{T, T-1} & \mathbf{R}^{T, T}
        \end{pmatrix},
\end{equation}
in which $\mathbf{R}^{i,j}\in\mathbb{R}^{p\times q}$ is a matrix and $\mathbf{R}^{i,j:k}$ denotes $[\mathbf{R}^{i,j},\cdots,\mathbf{R}^{i,k}]$ (a slight abuse of notation with aforementioned $A^{i,j:k}$).

\section{Related Work}
In this section, we highlight representative works in safe control for systems subject to model uncertainty or disturbances. The first and second parts discuss the value-based and policy-based approaches, respectively, for constructing safety filters with assured safety guarantees. The third part covers robust safe control methods that rely on solving the HJI PDE. The fourth part covers learning-based methods that enable scalable synthesis of robust safe controllers.
The pros and cons of these works are discussed in the Introduction, so we will not repeat them here.

\textbf{Value-based Safety Filters}: Online control corrections are enforced based on conditions on values and derivatives of robustified energy functions.
Breeden et al. propose a methodology for constructing robust control barrier functions for nonlinear control-affine systems with bounded disturbances \cite{breeden2023robust}.
Buch et al. present a robust control barrier function approach to handle sector-bounded uncertainties and utilize second-order cone programming for online optimization \cite{buch2021robust}.
Wei et al. propose a safety index synthesis method for systems with bounded state-dependent uncertainties and use convex semi-infinite programming to filter unsafe actions \cite{wei2022persistently}.
Wang et al. integrate the disturbance observer with control barrier functions for systems with external disturbances \cite{wang2023disturbance}.
Wei et al. introduce a robust safe control framework for systems with multi-modal uncertainties and derive the corresponding safety index synthesis method \cite{wei2023robust}.

\textbf{Policy-based Safety Filters}: Online control corrections are enforced based on whether a state can be safely guided back to terminal regions with some failsafe policy under worst-case disturbances. The failsafe policies are typically synthesized with RMPC or its variants.
Koller et al. utilize RMPC techniques for safe exploration of nonlinear systems with Gaussian process model uncertainties \cite{koller2018learning}.
Wabersich et al. propose a model predictive safety filter (MPSF) for disturbed nonlinear systems, replacing the objective of RMPC with the extent of safety filter intervention \cite{wabersich2021predictive}.
Leeman et al. propose an MPSF framework for linear systems with additive disturbances, utilizing system level synthesis techniques to reduce conservativeness \cite{leeman2023predictive}.
Bastani et al. propose a statistical model predictive shielding framework for nonlinear systems by training a recovery policy and enforcing sampling-based verification methods from RMPC \cite{bastani2021safe}.

\textbf{HJI PDE-based Robust Safe Control}: The HJI PDE is formulated and solved offline. The solution is then queried online to filter unsafe control inputs.
Fisac et al. propose a safety framework that utilizes the solution of HJI PDE and data-driven Bayesian inference to construct high-probability guarantees for the training process of learning-based algorithms \cite{fisac2018general}. Chen et al. propose a framework that allows for real-time motion planning with simplified system dynamics, in which an HJI PDE is formulated to obtain the tracking error bound \cite{chen2021fastrack}.

\textbf{Learning-based Robust Safe Control}: Reinforcement learning or supervised learning techniques are incorporated into the synthesis of robust safe control.
Dawson et al. develop a supervised learning approach to train a neural robust Lyapunov-barrier function for online control correction \cite{dawson2022safe}.
Liu et al. propose a robust training framework for safe RL under observational perturbations \cite{liu2022robustness}.
Li et al. propose an RL method for solving robust invariant sets with neural networks and use them for constrained policy optimization \cite{li2023robust}.

\section{Background}

\subsection{Hamilton–Jacobi Formulation for Robust Reach–Avoid Problems}

Hamilton-Jacobi formulation for robust reach-avoid problems provides a theoretical verification method for continuous-time nonlinear systems, ensuring they reach a designated target set while evading dangerous sets, even under the worst-case disturbances \cite{margellos2011hamilton}. The control inputs and disturbances are competing against each other in a differential game. The system dynamics is given by
\begin{equation}
    \dot{x}=f(x,u,a),\quad x\in \mathcal{S} , u\in \mathcal{U} , d\in \mathcal{D},
\end{equation}
in which $\mathcal{S}$ denotes the state space, $\mathcal{U}$ denotes the input space, $\mathcal{D}$ denotes disturbance space. The state flow map under the given control policy $\pi:\mathcal{S}\rightarrow\mathcal{U}$ and disturbance policy $\mu:\mathcal{S}\rightarrow\mathcal{D}$ is defined as
\begin{equation}
    \Phi _{f_{\pi ,\mu}}(x,t_1,t_2)=x+\int_{t_1}^{t_2}{f}(x(t),\pi (x(t)),\mu (x(t)))dt.
\end{equation}
The safety constraint for the system is $h(x)\leq0$. Therefore, the set of states that we want to avoid is given by $\mathcal{A}=\left\{x\mathrel{}\middle|\mathrel{}h(x)>0\right\}$. The target set that we want to reach is denoted by $\mathcal{R}=\left\{x\mathrel{}\middle|\mathrel{}l(x)\leq0\right\}$. The key is to define a value function
\begin{equation}
    \label{value function in Hamilton-Jacobi reach-avoid analysis}
    V(x,t)=\min _{\pi} \max _{\mu} \left\{V^{\pi,\mu}(x,t)\right\},
\end{equation}
in which $V^{\pi,\mu}(x,t)$ is given by
\begin{equation}
    \min _{\tau_1 \in[t, T]} \max \left\{l (\Phi _{f_{\pi ,\mu}}(x,\tau_1, t)), \max _{\tau_2 \in\left[t, \tau_1\right]} h(\Phi _{f_{\pi ,\mu}}(x,\tau_2, t))\right\}.
\end{equation}
$T$ denotes the time horizon.
The value function satisfies the following PDE \cite{margellos2011hamilton}:
\begin{equation}
    \label{HJ reach-avoid PDE}
    \begin{aligned}
        & \max \left\{h(x)-V(x, t), \frac{\partial V}{\partial t}(x, t)\right. \\
        & \left.+\min \left\{0, \max _{u \in \mathcal{U}} \min _{a \in \mathcal{A}} \frac{\partial V}{\partial x}(x, t) f(x, u, v)\right\}\right\}=0,
    \end{aligned}
\end{equation}
and the boundary condition is given by $V(x,T)=\max\left\{l(x),h(x)\right\}$.

Solving (\ref{HJ reach-avoid PDE}) requires discretizing the state space, which suffers from the curse of dimensionality. Hsu et al. propose an RL approach for solving reach-avoid problems \cite{hsu2021safety}. However, they assume that there is no disturbance in the system. It remains challenging to formulate (\ref{HJ reach-avoid PDE}) with disturbances in the RL setting.

\subsection{System Level Synthesis}

The framework of system level synthesis (SLS) establishes the equivalence of optimizing over feedback controllers and directly optimizing over system behaviors for linear time-varying (LTV) systems \cite{anderson2019system}. It has demonstrated efficacy in constrained robust and optimal controller design.

Consider a discrete-time LTV system with state $x\in\mathbb{R}^{n_x}$, input $u\in\mathbb{R}^{n_u}$, and disturbance $w\in\mathcal{B}_{\infty}^{n_x}\subset\mathbb{R}^{n_x}$:
\begin{equation}
    \label{LTV system}
    x_{k+1}=A_k x_k+B_k u_k+ \Sigma_k w_k,\quad k=0,1,\cdots,T-1.
\end{equation}
in which $A_k\in\mathbb{R}^{n_x\times n_x}$, $B_k\in\mathbb{R}^{n_x\times n_u}$, and the disturbance filter $\Sigma_k = {\rm diag}(\sigma_{k,1},\cdots,\sigma_{k,n_x})\in\mathbb{R}^{n_x\times n_x}$. The time horizon is denoted by $T$. The initial conditions are defined as $x_0=0$ and $u_0=0$.

Stacking the states, inputs, and disturbances of all timesteps as $\mathbf{x}=[x_1^{\top},\cdots,x_T^{\top}]^{\top}$, $\mathbf{u}=[u_1^{\top},\cdots,u_T^{\top}]^{\top}$, $\mathbf{w}=[w_0^{\top},\cdots,w_{T-1}^{\top}]^{\top}$, respectively,
the LTV system (\ref{LTV system}) can be written in the following compact form:
\begin{equation}
    \label{LTV system compact form}
    \mathbf{x}=\mathcal{Z} \mathbf{A} \mathbf{x}+\mathcal{Z} \mathbf{B} \mathbf{u}+\mathbf{\Sigma}\mathbf{w},
\end{equation}
in which $\mathbf{A}={\rm blkdiag}(A_1,\cdots,A_{T-1},0_{n_x,n_x})\in\mathbb{R}^{T n_x\times T n_x}$, $\mathbf{B}={\rm blkdiag}(B_1,\cdots,B_{T-1},0_{n_x,n_u})\in\mathbb{R}^{T n_x\times T n_u}$, $\mathbf{\Sigma}={\rm blkdiag}(\Sigma_0,\cdots,\Sigma_{T-1})\in\mathbb{R}^{T n_x\times T n_x}$, and $\mathcal{Z}\in\mathbb{R}^{T n_x\times T n_x}$ is a block downshift operator given by
\begin{equation}
    \mathcal{Z}=
    \begin{bmatrix}
        0_{n_{x}} & 0_{n_{x}} & \ldots & 0_{n_{x}} \\
        I_{n_{x}} & 0_{n_{x}} & \ldots & 0_{n_{x}} \\
        \vdots & \ddots & \ddots & \vdots \\
        0_{n_{x}} & \ldots & I_{n_{x}} & 0_{n_{x}}
    \end{bmatrix}.
\end{equation}

Assume that the control inputs follow the causal linear state-feedback design, i.e., $\mathbf{u}=\mathbf{K}\mathbf{x}$, in which $\mathbf{K} \in \mathcal{L}^{T, n_{u} \times n_{x}}$. The closed-loop system can be written as
\begin{equation}
    \label{system response for LTV}
    \begin{bmatrix}
        \mathbf{x} \\
        \mathbf{u}
    \end{bmatrix}
    =
    \begin{bmatrix}
        (I-\mathcal{Z} \mathbf{A}-\mathcal{Z} \mathbf{B} \mathbf{K})^{-1} \\
        \mathbf{K}(I-\mathcal{Z} \mathbf{A}-\mathcal{Z} \mathbf{B} \mathbf{K})^{-1}
    \end{bmatrix}
     \mathbf{\Sigma} \mathbf{w}=\begin{bmatrix} \mathbf{\Phi}_x \\ \mathbf{\Phi}_u \end{bmatrix} \mathbf{w}.
\end{equation}
$\mathbf{\Phi}_{x}\in\mathcal{L}^{T, n_{x} \times n_{x}}$ and $\mathbf{\Phi}_{u}\in\mathcal{L}^{T, n_{u} \times n_{x}}$ are referred to as system response. The following lemma allows for directly optimizing over system responses instead of linear state-feedback control inputs.

\begin{lemma}[system level synthesis \cite{chen2022robust}]
    \label{system level synthesis}
    For an LTV system (\ref{LTV system}) with a causal linear state-feedback controller $\mathbf{u}=\mathbf{K}\mathbf{x}$, $\mathbf{K} \in \mathcal{L}^{T, n_{u} \times n_{x}}$, the following properties hold.
    \begin{enumerate}
        \item The affine subspace defined by
            \begin{equation}
                \label{affine subspace for LTV system}
                \begin{bmatrix}
                    I-\mathcal{Z} \mathbf{A} & -\mathcal{Z} \mathbf{B}
                \end{bmatrix}
                    \begin{bmatrix} \mathbf{\Phi}_x \\ \mathbf{\Phi}_u \end{bmatrix}=\mathbf{\Sigma}
            \end{equation}
            parameterizes all possible system responses satisfying (\ref{system response for LTV}), where $\mathbf{\Phi}_{x}\in\mathcal{L}^{T, n_{x} \times n_{x}}$ and $\mathbf{\Phi}_{u}\in\mathcal{L}^{T, n_{u} \times n_{x}}$.
        \item For any $\left\{\mathbf{\Phi}_x,\mathbf{\Phi}_u\right\}$ satisfying (\ref{affine subspace for LTV system}), $\mathbf{\Phi}_x$ is invertible and the controller $\mathbf{K} =\mathbf{\Phi}_u \mathbf{\Phi}_x^{-1}$ achieves the system response (\ref{system response for LTV}).
    \end{enumerate}
\end{lemma}

\section{Problem Formulation and Algorithm Design}

\subsection{Problem Formulation and Motivation}

In this work, we consider discrete-time nonlinear systems with state- and action-dependent perturbations:
\begin{equation}
    \label{nonlinear system with disturbance}
    \begin{aligned}
        x_{k+1}&=F(x_k,u_k,d_k)\\&=f(x_k,u_k)+g(x_k,u_k)d_k.
    \end{aligned}
\end{equation}
in which state $x\in\mathbb{R}^{n_x}$, input $u\in\mathbb{R}^{n_u}$, and disturbance $d\in\mathbb{R}^{n_d}$. Note that $g(x,u)d$ can also be viewed as model uncertainty. The functions $f:\mathcal{R}^{n_x}\times \mathcal{R}^{n_u}\rightarrow \mathcal{R}^{n_x}$ and $g:\mathcal{R}^{n_x}\times \mathcal{R}^{n_u}\rightarrow \mathcal{R}^{n_x\times n_d}$ are assumed to be smooth.

The state constraint set $\mathcal{X}$ is assumed to be the following polytopic form:
\begin{equation}
    \label{state constraint set}
    \mathcal{X}=\left\{x\in\mathbb{R}^{n_x}\mathrel{}\middle|\mathrel{} H_x x\leq h_x\right\}.
\end{equation}
The input constraint set $\mathcal{U}$ is assumed to be the following polytopic form:
\begin{equation}
    \label{input constraint set}
    \mathcal{U}=\left\{u\in\mathbb{R}^{n_u}\mathrel{}\middle|\mathrel{} H_u u\leq h_u\right\}.
\end{equation}
The disturbance set $\mathcal{D}$ is assumed to be the following polytopic form:
\begin{equation}
    \label{disturbance set}
    \mathcal{D}=\left\{d\in\mathbb{R}^{n_d}\mathrel{}\middle|\mathrel{} H_d d\leq h_d\right\}.
\end{equation}

We also assume having access to a terminal safe set $\mathcal{R}\subset \mathcal{X}$ defined as
\begin{equation}
    \label{terminal safe set}
    \mathcal{R}=\left\{x\in\mathbb{R}^{n_x}\mathrel{}\middle|\mathrel{} R_x x\leq r_x\right\},
\end{equation}
and a corresponding terminal safe controller $\pi_{{\rm terminal}}: \mathcal{X}\rightarrow\mathcal{U}$ that can keep the system safe for any initial state $x_0\in\mathcal{R}$, i.e.,
\begin{equation}
    \label{terminal safe controller requirement}
    \begin{aligned}
        &x_k \in \mathcal{X},\quad \forall k\geq 0,\quad\forall d_k\in\mathcal{D},\\&
        \begin{aligned}
            \text{s.t.}\quad&x_0\in \mathcal{R}, \quad u_k=\pi_{{\rm terminal}}(x_k),\\& x_{k+1}=f(x_k,u_k)+g(x_k,u_k)d_k.
        \end{aligned}
    \end{aligned}
\end{equation}
Note that the assumption on terminal safe sets is standard in MPC-related settings. For example, terminal safe regions can be obtained by linearizing the system around its equilibriums \cite{chen1998quasi, althoff2008reachability}.
Also note that here we relaxed the requirement of the terminal safe controller $\pi_{{\rm terminal}}$ in the sense that it does not need to render the terminal set $\mathcal{R}$ forward invariant, i.e., $x_k \in \mathcal{X}$ instead of $x_k \in \mathcal{R}$ in (\ref{terminal safe controller requirement}).

Safety filters ensure that the system trajectories are always satisfying safety constraints, by modifying dangerous actions that can lead to unsafe states \cite{liu2014control, wabersich2023data}. They can be combined with any nominal controllers, such as human-designed control laws and RL policies. A safety filter can be viewed as a function $F_{\rm filter}:\mathbb{R}^{n_x}\times\mathbb{R}^{n_u}\rightarrow\mathcal{U}$, such that for initial states $x_0\in \mathcal{S}_{{\rm safe}}$, the safety filter always returns safe actions, i.e.,
\begin{equation}
    \label{safety filter problem formulation}
    \begin{aligned}
        &x_k \in \mathcal{S}_{{\rm safe}},\quad \forall k\geq 0, \quad\forall d_k\in\mathcal{D},\\&
        \begin{aligned}
            \text{s.t.}\quad&x_0\in \mathcal{S}_{{\rm safe}},\quad u_k=F_{\rm filter}(x_k, u_{{\rm nom}})\\& x_{k+1}=f(x_k,u_k)+g(x_k,u_k)d_k,
        \end{aligned}
    \end{aligned}
\end{equation}
in which $u_{{\rm nom}}$ denotes the outputs of a nominal controller. The safe set $\mathcal{S}_{{\rm safe}} \subseteq \mathcal{X}$ represents the collection of states such that (\ref{safety filter problem formulation}) is feasible. It is a robust invariant set (RIS) that is (implicitly) determined by the safety filter. (Technically it is a robust control invariant set. We omit the word control in this article for simplicity.)

Conservativeness is the key evaluation metric for a safety filter, which is mainly reflected by the size of the corresponding safe set $\mathcal{S}_{{\rm safe}}$. Larger safe set sizes indicate wider operational ranges for the system and its controllers, which result in a significant boost on controller performance. Therefore, the primary design principle for safety filter $F_{\rm filter}$ is to maximize its safe set. The theoretical maximal safe set that a safety filter can obtain is the maximal RIS of the system, which is very difficult to obtain in general (having to solve the HJI PDE).

Existing methods for constructing safety filters offer different trade-offs between conservativeness and online computational burden. The policy-based approach (e.g., RMPC) typically has larger safe sets than that of value-based approach (e.g., robust control barrier function), but requires solving non-convex optimization problems that can be computationally prohibitive.

In this work, we leverage RL tools to design a safety filter $F_{\rm filter}$ that requires significantly less computational effort than policy-based filters, while still retaining their benefits, i.e., large safe sets and rigorous safety guarantees.

\subsection{Algorithm Design}

We decompose the (maximal) RIS into two parts: a target set that aligns with terminal region design of RMPC, and a reach-avoid set that encapsulates the rest of RIS. The latter is derived using deep RL techniques, the specifics of which we will elaborate on in the following section. Three neural networks are trained concurrently: a reach-avoid policy network, a disturbance network, and a reach-avoid value network. The advantage of the RL approach lies in its scalability and suitability for complex systems. However, this method sacrifices the theoretical guarantee for safety. For instance, the learned reach-avoid policy network cannot guarantee to safely guide the system back to the target set under worst-case disturbances, even if the reach-avoid value network indicates a negative value—though, empirically, this is often the case with high probability.

To reconstruct safety guarantees, we propose a second-order cone programming (SOCP) approach for online safety verification, leveraging the tools of system level synthesis. Specifically, given a certain state, the learned reach-avoid policy network is utilized to generate a nominal trajectory. Using the SOCP method, we compute the worst-case reach-avoid value for trajectories that might deviate from the nominal one due to disturbances, taking advantage of additional linear state-feedback controllers. If this worst-case value is negative, then the system is guaranteed to be safe. The idea of this online verification is similar to RMPC. The key difference is that the nominal trajectory is directly generated by the learned reach-avoid policy network, while in RMPC approaches it is part of the decision variables to be optimized, resulting in non-convex optimization problems, which is much more computationally expensive.

\begin{figure*}[htbp]
    \centering
    \includegraphics[width=5.5in]{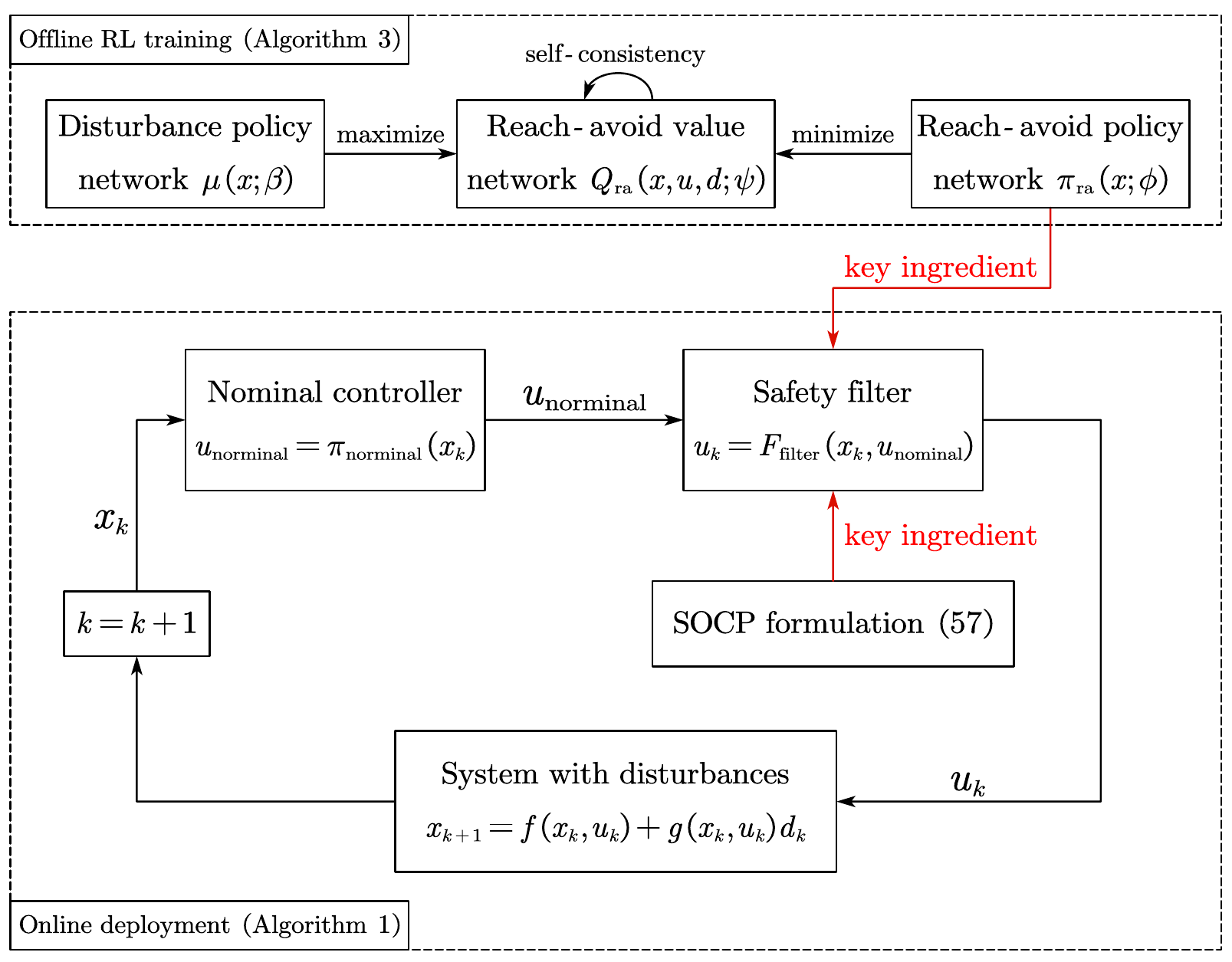}
    \caption{Illustration of the proposed methods, including offline RL training and online SOCP-based safety filtering.}
    \label{overall design}
\end{figure*}

\begin{algorithm}[ht]
    \label{Safety filter design}
    \caption{Online deployment of the proposed safety filter}
    \KwIn{learned reach-avoid policy $\pi_{\rm ra}$, potentially unsafe nominal controller $\pi_{\rm nominal}$, terminal safe controller $\pi_{\rm terminal}$.}

    Observe the current state $\overline{x}$ and get the nominal input $\overline{u}=\pi_{\rm nominal}(\overline{x})$. \label{start of first computation}

    Generate the nominal trajectory $\mathcal{T}$ using $\pi_{\rm ra}$ with $z_0=\overline{x}$ and $v_0=\overline{u}$, which is detailed in (\ref{generate nominal trajectory}).

    Solve the SOCP problem (\ref{SOCP formulation}). Obtain $V_{\rm ra}^*$, $\mathbf{\Phi}_u^*$ and ${\mathbf{\Phi}_x^*}$. (We assume $V_{\rm ra}^*\leq0$ for the first computation.) \label{end of first computation}

    Update $\mathbf{K}$ with $\mathbf{K}=\mathbf{\Phi}_u^* {\mathbf{\Phi}_x^*}^{-1}$. \label{update K}

    Execute $u=v_0$ at current state.

    \For{$k = 1, 2, \cdots$}{

        Observe the current state $\overline{x}$ and get the nominal input $\overline{u}=\pi_{\rm nominal}(\overline{x})$.

        Generate the nominal trajectory $\mathcal{T}$ using $\pi_{\rm ra}$ with $z_0=\overline{x}$ and $v_0=\overline{u}$, which is detailed in (\ref{generate nominal trajectory}). \label{generate nominal trajectory on-the-fly}

        Solve the SOCP problem (\ref{SOCP formulation}). Otain $V_{\rm ra}^*$, $\mathbf{\Phi}_u^*$ and ${\mathbf{\Phi}_x^*}$. \label{solve SOCP on-the-fly}

        \uIf{$V_{\rm ra}^*\leq 0$}{
            Reset $k=0$.

            Update the nominal trajectory $\mathcal{T}$ with the results obtained in line \ref{generate nominal trajectory on-the-fly}.

            Go to line \ref{update K} with system responses $\mathbf{\Phi}_u^*$ and ${\mathbf{\Phi}_x^*}$ obtained in line \ref{solve SOCP on-the-fly}.
        }
        \Else{

            Abandon the results of line \ref{generate nominal trajectory on-the-fly} and \ref{solve SOCP on-the-fly}.

            \uIf{$k\leq T$}{
                Execute $u=v_k+\Sigma_{j=1}^k{\mathbf{K}^{k,j}\Delta x_j}$ at current state.
            }
            \Else{
                Execute $u = \pi_{\rm terminal}(x_k)$ at current state.
            }
        }
    }
\end{algorithm}

The execution procedure of the proposed safety filter is summarized in Algorithm \ref{Safety filter design}. We assume that the first-time computation of the SOCP problem (from line \ref{start of first computation} to \ref{end of first computation}) yields a $V_{\rm ra}^*\leq0$, which is similar to the initial feasibility assumption for RMPC. We will show that the proposed safety filter guarantees persistent robust constraint satisfaction in Theorem \ref{persistent safety guarantee}.

We also provide a diagram that illustrates the overall design of the proposed methods in Fig. \ref{overall design}. The two key ingredients of our approach are the neural network reach-avoid policy trained offline and the SOCP formulation for online verification. Note that the SOCP formulation (\ref{SOCP formulation}) provides the worst-case reach-avoid value $V_{\rm ra}^*$ and the optimal linear feedback stabilizing controller $\mathbf{K}$ simultaneously. The latter will be used for at most $T$ following timesteps, if the SOCP results at these steps all have positive $V_{\rm ra}^*$. Therefore the timestep index $k$ is actually entangled with the safety filter $F_{\rm filter}$, which is omitted in Fig. \ref{overall design} for simplicity.

\section{RL for Robust Reach-Avoid Problems}

In this section, we present an RL approach for solving the robust reach-avoid problems. Note that the Hamilton-Jacobi formulation (\ref{HJ reach-avoid PDE}) is for continuous-time systems with a finite horizon. Here we address discrete-time systems under infinite horizon.

We consider a general setting for discrete-time systems with disturbances,
\begin{equation}
    \label{general nonlinear system with disturbance}
    x_{k+1}=F(x_k,u_k,d_k), \quad x\in \mathcal{S} , u\in \mathcal{U} , d\in \mathcal{D}.
\end{equation}
The control input $u$ follows policy $\pi:\mathcal{S}\rightarrow\mathcal{U}$ and disturbance $d$ follows policy $\mu:\mathcal{S}\rightarrow\mathcal{D}$. The safety constraint is $h(x)\leq0$. The target constraint is $l(x)\leq0$. The avoiding set $\mathcal{A}$ is defined as $\mathcal{A}=\left\{x\mathrel{}\middle|\mathrel{}h(x)>0\right\}$. The target set $\mathcal{R}$ is defined as $\mathcal{R}=\left\{x\mathrel{}\middle|\mathrel{}l(x)\leq0\right\}$.

The above settings cover the problem formulation in the previous section. (\ref{nonlinear system with disturbance}) is a special case of (\ref{general nonlinear system with disturbance}). (\ref{state constraint set}) implies that $h(x) = \max_i\left\{H_x^{i,:} x - h_x^i\right\}$ and (\ref{terminal safe set}) implies that $l(x) = \max_i\left\{R_x^{i,:} x - r_x^i\right\}$.

We define the following reach-avoid value function. Our approach is inspired by \cite{hsu2021safety, li2023robust}.

\begin{definition}[reach-avoid value functions]
    \hfill
    \label{reach-avoid value functions}
    \begin{enumerate}
        \item Let $\mathcal{T}(x,\tau)$ denote a trajectory of the system (\ref{general nonlinear system with disturbance}) starting at $x$ with a finite horizon $\tau$, i.e., $\mathcal{T}(x,\tau)=\left\{x_0=x,u_0,d_0,x_1,u_1,d_1,\cdots,x_\tau,u_\tau,d_\tau\right\}$. The reach-avoid value of a finite trajectory $\mathcal{T}(x,\tau)$ is defined as
        \begin{equation}
            V_{\rm ra}(\mathcal{T}(x,\tau)) = \max\left\{\max_{0\leq t\leq\tau,t\in \mathbb{N}}\left\{h(x_t)\right\}, l(x_{\tau})\right\}.
        \end{equation}
        \item The reach-avoid value function of a protagonist policy $\pi$ and an adversary policy $\mu$ is defined as
        \begin{equation}
            \label{definition of V_RA}
            V_{\rm ra}^{\pi,\mu}(x) = \min_{\tau\in\mathbb{N}}\left\{V_{\rm ra}(\mathcal{T}(x,\tau))\right\},
        \end{equation}
        in which the trajectory $\mathcal{T}(x,\tau)$ is obtained by starting at $x_0=x$ and choosing $u_t=\pi(x_t), a_t=\mu(x_t)$ for $0\leq t\leq \tau$.
        \item The reach-avoid value function of a protagonist policy $\pi$ is defined as
        \begin{equation}
            \label{definition of V_RA pi}
            V_{\rm ra}^{\pi}(x) = \max _{\mu} \left\{V_{\rm ra}^{\pi,\mu}(x)\right\}.
        \end{equation}
        \item The optimal reach-avoid value function is defined as
        \begin{equation}
            \label{definition of optimal V_RA}
            V_{\rm ra}^*(x) = \min _{\pi} \left\{V_{\rm ra}^{\pi}(x)\right\}.
        \end{equation}
    \end{enumerate}
\end{definition}

The reach-avoid value of a finite trajectory reflects whether the requirement of reaching $\mathcal{R}$ and avoiding $\mathcal{A}$ is fulfilled. For the reaching of $R$, the target constraint is imposed on the last timestep $\tau$. For the avoiding of $A$, the safety constraint is imposed on every timestep $0\leq t\leq\tau$. If $V_{\rm ra}(\mathcal{T}(x,\tau))\leq0$, the requirement of reaching and avoiding are both achieved in this finite trajectory $\mathcal{T}(x,\tau)$. For a pair of policy $\pi$ and $\mu$, the reach-avoid value function chooses the optimal horizon $\tau$ such that the reach-avoid value of the corresponding trajectory $\mathcal{T}(x,\tau)$ is minimum, which is the best performance that can be achieved on the goal of reach-avoid. The reach-avoid value function of a policy $\pi$ reflects its capability of reach-avoid under worst-case disturbances. The optimal reach-avoid value function is associated with the best policy for the goal of reach-avoid under worst-case disturbances. We further define the reach-avoid sets as follows.

\begin{definition}[reach-avoid sets]
    \hfill
    \begin{enumerate}
        \item The reach-avoid set for a policy $\pi$ is the zero-sublevel set of its reach-avoid value function, i.e.,
        \begin{equation}
            S_{\rm ra}^{\pi}=\left\{x \mid V_{\rm ra}^{\pi}(x) \leq 0\right\}.
        \end{equation}
        \item The maximal reach-avoid set is the zero-sublevel set of the optimal reach-avoid value function, i.e.,
        \begin{equation}
            S_{\rm ra}^*=\left\{x \mid V_{\rm ra}^*(x) \leq 0\right\}.
        \end{equation}
    \end{enumerate}
\end{definition}

For states outside the maximal reach-avoid set $S_{\rm ra}^*$, there is no policy that can drive the system to the target set while satisfying the safety constraints, under worst-case disturbances.

Since the reach-avoid value functions are defined on infinite horizon, they naturally hold a recursive structure with dynamic programming, which is referred to as the self-consistency condition, just like the common value functions in the standard RL setting.
We have the following theorem.

\begin{theorem}[self-consistency conditions for reach-avoid value functions]
    \label{self-consistency conditions for reach-avoid value functions}
    Suppose $x^{\prime}$ is the successive state of $x$. The following self-consistency conditions hold for reach-avoid value functions, i.e.,
    \begin{equation}
        \label{self-consistency condition of V_RA pi,mu}
        V_{\rm ra}^{\pi ,\mu}(x)=\max \left\{h(x), \min\left\{l(x),V_{\rm ra}^{\pi ,\mu}\left( x^{\prime} \right)\right\} \right\},
    \end{equation}
    \begin{equation}
        \label{self-consistency condition of V_RA pi}
        V_{\rm ra}^{\pi}(x)=\max \left\{h(x), \min\left\{l(x), \max_{d\in \mathcal{D}}V_{\rm ra}^{\pi}\left( x^{\prime} \right)\right\} \right\},
    \end{equation}
    \begin{equation}
        \label{self-consistency condition of V_RA}
        V_{\rm ra}^{*}(x)=\max \left\{h(x), \min\left\{l(x), \min_{u\in \mathcal{U}}\max_{d\in \mathcal{D}}V_{\rm ra}^{*}\left( x^{\prime} \right)\right\} \right\}.
    \end{equation}
\end{theorem}

\begin{proof}
    We first prove (\ref{self-consistency condition of V_RA pi,mu}).
    Let $x^{\prime}$ denote the successive state of $x$ when the system is driven by $\pi$ and $\mu$, i.e., $x^{\prime}=F(x,\pi(x),\mu(x))$.
    We have
    \begin{equation}
        \begin{aligned}
            &V_{\rm ra}^{\pi ,\mu}(x)=\min_{\tau\geq0} \left\{ V_{\rm ra}(\mathcal{T}(x,\tau)) \right\}\\
            &=\min \left\{ V_{\rm ra}(\mathcal{T}(x,0)),\min_{\tau\ge 1} \left\{ V_{\rm ra}(\mathcal{T}(x,\tau))\right\}\right\}\\
            &=\min \left\{ V_{\rm ra}(\mathcal{T}(x,0)), \max\left\{h(x), \min_{\tau\ge 0} \left\{ V_{\rm ra}(\mathcal{T}(x^{\prime},\tau)) \right\} \right\}\right\}\\
            &=\min \left\{ \max\left\{l(x),h(x)\right\}, \max\left\{h(x), V_{\rm ra}^{\pi ,\mu}(x^{\prime}) \right\}\right\}\\
            &=\max \left\{h(x), \min\left\{l(x),V_{\rm ra}^{\pi ,\mu}\left( x^{\prime} \right)\right\} \right\},
        \end{aligned}
    \end{equation}
    in which the second equality follows from manipulating the index of subscript $\tau$ based on the definition of $V_{\rm ra}(\mathcal{T}(x,\tau))$, and the last equality follows from the relationship $\min \{\max \{a, b\}, \max \{b, c\}\}=\max \{b, \min \{a, c\}\}$.

    (\ref{self-consistency condition of V_RA pi}) and (\ref{self-consistency condition of V_RA}) can be obtained by further considering Bellman's principle of optimality.
\end{proof}

The self-consistency conditions for reach-avoid value functions in Theorem \ref{self-consistency conditions for reach-avoid value functions} are not contraction mappings, which is different from the standard RL setting. Therefore, they can not be directly utilized for policy iteration. Following the method in \cite{hsu2021safety}, we introduce a discount factor $\gamma_{\rm ra}$ into the formulations and obtain three contractive self-consistency operators.

\begin{definition}[reach-avoid self-consistency operators]
    \label{reach-avoid self-consistency operators}
    Suppose $0<\gamma_{\rm ra}<1$. The reach-avoid self-consistency operators are defined as
    \begin{equation}
        \begin{aligned}
            &\left[ T_{\rm ra}^{\pi ,\mu}(V_{\rm ra}) \right] (x)=(1-\gamma_{\rm ra})\max\left\{l(x),h(x)\right\}\\&+\gamma_{\rm ra}\max \left\{h(x), \min\left\{l(x),V_{\rm ra}\left( x^{\prime} \right)\right\} \right\},
        \end{aligned}
    \end{equation}
    \begin{equation}
        \begin{aligned}
            &\left[ T_{\rm ra}^{\pi}(V_{\rm ra}) \right] (x)=(1-\gamma_{\rm ra})\max\left\{l(x),h(x)\right\}\\&+\gamma_{\rm ra}\max \left\{h(x), \min\left\{l(x), \max_{d\in \mathcal{D}}V_{\rm ra}\left( x^{\prime} \right)\right\} \right\},
        \end{aligned}
    \end{equation}
    \begin{equation}
        \begin{aligned}
            &\left[ T_{\rm ra}(V_{\rm ra}) \right] (x)=(1-\gamma_{\rm ra})\max\left\{l(x),h(x)\right\}\\&+\gamma_{\rm ra}\max \left\{h(x), \min\left\{l(x), \min_{u\in \mathcal{U}}\max_{d\in \mathcal{D}}V_{\rm ra}\left( x^{\prime} \right)\right\} \right\}.
        \end{aligned}
    \end{equation}
\end{definition}

\begin{theorem}[monotone contraction]
    \label{monotone contraction of self-consistency operators}
    Let $T$ denote any of ${T_{\rm ra}^{\pi ,\mu},T_{\rm ra}^{\pi},T_{\rm ra}}$.
    \begin{enumerate}
        \item Given any $V_{\rm ra},\widetilde{V}_{\rm ra}$, we have
        \begin{equation}
            \label{contraction}
            \left\|T(V_{\rm ra})-T(\widetilde{V}_{\rm ra})\right\|_{\infty} \leq \gamma_{\rm ra}\left\|V_{\rm ra}-\widetilde{V}_{\rm ra}\right\|_{\infty}.
        \end{equation}
        \item Suppose $V_{\rm ra}(x)\geq \widetilde{V}_{\rm ra}(x)$ holds for any $x\in\mathcal{S}$. Then we have
        \begin{equation}
            \label{monotone}
            [T(V_{\rm ra})](x)\geq[T(\widetilde{V}_{\rm ra})](x),\quad \forall x \in\mathcal{S}.
        \end{equation}
    \end{enumerate}
\end{theorem}

\begin{proof}
    We only prove the monotone contraction for $T_{\rm ra}$, while the proof for $T_{\rm ra}^{\pi}$ and $T_{\rm ra}^{\pi,\mu}$ is similar.
    We have
    \begin{equation}
        \begin{aligned}
            &[T_{\rm ra}( V_{\rm ra} )](x) - [T_{\rm ra}( \widetilde{V}_{\rm ra} )](x)
            \\=&\gamma_{\rm ra} \max \left\{h(x), \min\left\{l(x), \min_{u\in \mathcal{U}}\max_{d\in \mathcal{D}}V_{\rm ra}\left( x^{\prime} \right)\right\} \right\}\\&-\gamma_{\rm ra} \max \left\{h(x), \min\left\{l(x), \min_{u\in \mathcal{U}}\max_{d\in \mathcal{D}}\widetilde{V}_{\rm ra}\left( x^{\prime} \right)\right\} \right\}.
        \end{aligned}
    \end{equation}
    We have
    \begin{equation}
        \label{proving contraction}
        \begin{aligned}
            &\left\|T_{\rm ra}( V_{\rm ra} ) -T_{\rm ra}( \widetilde{V}_{\rm ra} )\right\|_{\infty}\\
            \leq& \gamma_{\rm ra} \max\limits_{u\in \mathcal{U}} \max\limits_{d\in \mathcal{D}} \left| V_{\rm ra}(f(x,u,a))-\widetilde{V}_{\rm ra}(f(x,u,a)) \right|\\
            \leq& \gamma_{\rm ra} \left\| V_{\rm ra}-\widetilde{V}_{\rm ra} \right\| _{\infty}.
        \end{aligned}
    \end{equation}
    The first inequality in (\ref{proving contraction}) follows from the relationship
    \begin{equation}
        \begin{aligned}
            &\left|\min _x \max _y f(x, y)-\min _x \max _y g(x, y)\right| \\&\leq \max _x \max _y\left|f(x, y)-g(x, y)\right|.
        \end{aligned}
    \end{equation}

    Since $\max$ and $\min$ operations are monotone, the monotonicity of $T_{\rm ra}$ is obvious.
\end{proof}

In the following proposition, we show that if $\gamma_{\rm ra}$ is sufficiently close to 1, the fixed points of the self-consistency operators approach the original reach-avoid value functions.

\begin{proposition}
    \label{approximation}
    Let $T$ denote any of ${T_{\rm ra}^{\pi ,\mu},T_{\rm ra}^{\pi},T_{\rm ra}}$, and $V_{\rm ra}^d$ denote the fixed point of operator $T$, i.e., $T(V_{\rm ra}^d)=V_{\rm ra}^d$. Let $V_{\rm ra}$ denote the corresponding original reach-avoid value function in Definition \ref{reach-avoid value functions}. Then we have $\lim _{\gamma_{\rm ra} \rightarrow 1}V_{\rm ra}^d(x)=V_{\rm ra}(x)$.
\end{proposition}

\begin{proof}
    The key is to define a discounted version of reach-avoid values. The original definition of $V_{\rm ra}^{\pi,\mu}(x)$ is given by $V_{\rm ra}^{\pi,\mu}(x) = \min_{\tau\in\mathbb{N}}\left\{V_{\rm ra}(\mathcal{T}(x,\tau))\right\}$. Let $\left\{x_t\right\}$ denote the trajectory of the system when it is driven by $\pi$ and $\mu$, starting from $x_0=x$. We define the discounted reach-avoid value as
    \begin{equation}
        \label{explicit form of safety value}
        \begin{aligned}
            D(x)&=(1-\gamma_{\rm ra}) \max \left\{l(x_0), h(x_0)\right\}+\gamma_{\rm ra} \max \left\{h(x_0), \right.\\&\min \left\{l(x_0),\right.
            \left.(1-\gamma_{\rm ra}) \max \left\{l(x_1), h(x_1)\right\}\right.\\&\left.+\gamma_{\rm ra} \max \left\{h(x_1), \min \left\{l(x_1), \cdots\right\}\right\}\right\},
        \end{aligned}
    \end{equation}
    It can be verified that $D$ is the explicit form of the fixed point of $T_{\rm ra}^{\pi ,\mu}$, i.e., $D=T_{\rm ra}^{\pi ,\mu}(D)$. Forcing $\gamma_{\rm ra}$ to be sufficiently close to 1, we have $\lim _{\gamma_{\rm ra} \rightarrow 1}D(x)=V_{\rm ra}^{\pi,\mu}(x)$. The proof is similar for $T_{\rm ra}^{\pi}$ and $T_{\rm ra}$.
\end{proof}

Since the reach-avoid self-consistency operators are monotone contractions, we can utilize policy iteration to solve for the reach-avoid sets. The policy iteration algorithm is shown in Algorithm \ref{Policy iteration algorithm}, in which we iterate between policy evaluation and policy improvement. The former evaluates the reach-avoid value function of the current policy $\pi$, by finding the worst-case disturbances for $\pi$. The latter finds a better policy on the performance of reach-avoid, by solving a maximin problem. The convergence of the proposed policy iteration is shown in Theorem \ref{monotone convergence of policy iteration}.

\begin{algorithm}[ht]
    \label{Policy iteration algorithm}
    \caption{Policy iteration for robust reach-avoid problems}
    \KwIn{initial policy $\pi_0$.}
    \For{each iteration $k$}{
        \textit{(policy evaluation)}
        
        Solve for $V_{\rm ra}^{\pi_k}$ such that $V_{\rm ra}^{\pi_k}=T_{\rm ra}^{\pi_k}(V_{\rm ra}^{\pi_k})$.
        
        \textit{(policy improvement)}

        $\pi_{k+1}=\underset{\pi}{\mathrm{arg}\min}\left\{ T_{\rm ra}^{\pi}\left( V_{\rm ra}^{\pi_k} \right) \right\}$.
    }
\end{algorithm}

\begin{theorem}[monotone convergence of policy iteration]
    \label{monotone convergence of policy iteration}
    The sequence $\left\{V_{\rm ra}^{\pi_k}\right\}$ generated by Algorithm \ref{Policy iteration algorithm} converges monotonically to the fixed point $V_{\rm ra}^*$ of $T_{\rm ra}$, i.e., $T_{\rm ra}(V_{\rm ra}^*)=V_{\rm ra}^*$.
\end{theorem}

\begin{proof}
    The key is to exploit the monotonicity and contraction of $T_{\rm ra}^{\pi}$ and $T_{\rm ra}$. First we establish the following recursive relationship:
    \begin{equation}
        \label{recursion for monotone convergence}
        V_{\rm ra}^{\pi _k}\ge T_{\rm ra}( V_{\rm ra}^{\pi _k} ) \ge V_{\rm ra}^{\pi _{k+1}}\ge V_{\rm ra}^{*}.
    \end{equation}
    With the definition of $T_{\rm ra}$ and policy improvement, we have
    \begin{equation}
        T_{\rm ra} ( V_{\rm ra}^{\pi_k})=T_{\rm ra}^{\pi_{k+1}}(V_{\rm ra}^{\pi_k}) \leq T_{\rm ra}^{\pi_k}(V_{\rm ra}^{\pi_k})=V_{\rm ra}^{\pi_k}.
    \end{equation}
    Utilizing the monotonicity and contraction of $T_{\rm ra}^{\pi_{k+1}}$, together with $V_{\rm ra}^{\pi_k} \geq T_{\rm ra}^{\pi_{k+1}}(V_{\rm ra}^{\pi_k})$, we have
    \begin{equation}
        \begin{aligned}
            V_{\rm ra}^{\pi_k} &\geq T_{\rm ra}(V_{\rm ra}^{\pi_k})=T_{\rm ra}^{\pi_{k+1}}(V_{\rm ra}^{\pi_k})\\& \geq(T_{\rm ra}^{\pi_{k+1}})^{\infty}(V_{\rm ra}^{\pi_k})=V_{\rm ra}^{\pi_{k+1}}.
        \end{aligned}
    \end{equation}
    Utilizing the monotonicity and contraction of $T_{\rm ra}$, and $T_{\rm ra}(V_{\rm ra}^{\pi_{k+1}}) \leq T_{\rm ra}^{\pi_{k+1}}(V_{\rm ra}^{\pi_{k+1}})=V_{\rm ra}^{\pi_{k+1}}$, we have
    \begin{equation}
        V_{\rm ra}^*=(T_{\rm ra})^{\infty}(V_{\rm ra}^{\pi_{k+1}}) \leq \cdots \leq V_{\rm ra}^{\pi_{k+1}}.
    \end{equation}
    So the recursive relationship (\ref{recursion for monotone convergence}) holds. The sequence $\left\{V_{\rm ra}^{\pi_k}\right\}$ generated by policy iteration is monotonically decreasing and bounded by $V_{\rm ra}^{\pi_0}\geq V_{\rm ra}^{\pi_k}\geq V_{\rm ra}^*$, so it converges. The convergence implies that $V_{\rm ra}^{\pi_k}=T_{\rm ra}(V_{\rm ra}^{\pi_k})=V_{\rm ra}^{\pi_{k+1}}$, which means that the convergence point is the fixed point of $T_{\rm ra}$.
\end{proof}

The policy iteration approach enjoys rigorous convergence guarantee and the policy is guaranteed to fulfill the goal of reach-avoid for states inside its reach-avoid set. However, it requires state-wise computation and is computationally prohibitive for continuous state spaces. Still, it paves the way for the design of an actor-critic deep RL algorithm, which is of high scalability and applicable to complex systems with continuous state and action spaces.

We utilize the action reach-avoid values for the actor-critic algorithm. The relationship between the action reach-avoid values and the reach-avoid values is similar to that of common action values and common values in the standard RL setting. The algorithm includes three neural networks: the reach-avoid policy network $\pi_{\rm ra}(x;\phi)$, the disturbance network $\mu(x;\beta)$ and the reach-avoid value network $Q_{\rm ra}(x,u,d;\psi)$. $\phi$, $\beta$ and $\psi$ denote the network parameters.
For a set $\mathcal{D}$ of collected samples, the loss function for $Q_{\rm ra}(x,u,a;\psi)$ is
\begin{equation}
    L_{Q_{\rm ra}}(\psi)=\mathbb{E}_{(x,u,d,h,l,x^{\prime}) \sim \mathcal{D}}\left\{\left(Q_{\rm ra}(x,u,d ; \psi)-\hat{Q}_{\rm ra}\right)^2\right\},
\end{equation}
where
\begin{equation}
    \begin{aligned}
        &\hat{Q}_{\rm ra} = (1-\gamma_{\rm ra})\max\left\{l(x),h(x)\right\}\\&+\gamma_{\rm ra}\max \left\{h(x), \min\left\{l(x),Q_{\rm ra}\left(x^{\prime}, u^{\prime}, d^{\prime};\bar{\psi}\right)\right\} \right\},
    \end{aligned}
\end{equation}
in which $u^{\prime}=\pi_{\rm ra}(x^{\prime};\phi )$ and $d^{\prime}=\mu(x^{\prime};\beta)$. $\bar{\psi}$ denotes target network parameters for $Q_{\rm ra}$, which follows the standard design in RL algorithms. This loss function is based on the reach-avoid self-consistency operators.
$\pi_{\rm ra}(x;\phi)$ aims to minimize the reach-avoid value. Its loss function is
\begin{equation}
    L_{\pi_{\rm ra}}(\phi)=\mathbb{E}_{x \sim \mathcal{D}}\left\{Q_{\rm ra}\left(x,\pi_{\rm ra}(x;\phi ),\mu(x;\beta) ; \psi\right)\right\}.
\end{equation}
$\mu(x;\beta)$ aims to maximize the reach-avoid value. Its loss function is
\begin{equation}
    L_{\mu}(\beta)=-\mathbb{E}_{x \sim \mathcal{D}}\left\{Q_{\rm ra}\left(x,\pi_{\rm ra}(x;\phi ),\mu(x;\beta) ; \psi\right)\right\}.
\end{equation}
The overall algorithm is summarized in Algorithm \ref{deep RL Policy iteration algorithm}.

\begin{algorithm}
    \label{deep RL Policy iteration algorithm}
    \caption{Actor-critic algorithm for robust reach-avoid problems}
    \KwIn{network parameters $\psi$, $\phi$ and $\beta$, target network parameter $\bar{\psi}\leftarrow\psi$, learning rate $\eta$, target smoothing coefficient $\tau$, replay buffer $\mathcal{D}\leftarrow \varnothing$.}
    \For{each iteration}{
        \For{each system step}{
            Sample control input $u_t\sim \pi_{\rm ra}(x_t;\phi)$;
            
            Sample disturbance $d_t\sim \mu(x_t;\beta)$;

            Observe next state $x_{t+1}$, safety constraint value $h_t$, target constraint value $l_t$;

            Store transition $\mathcal{D} \leftarrow \mathcal{D} \cup\left\{\left(x_t, u_t, d_t, h_t, l_t, x_{t+1}\right)\right\}$.
        }

        \For{each gradient step}{
            Sample a batch of data from $\mathcal{D}$;

            Update reach-avoid value function $\psi \leftarrow \psi-\eta \nabla_\psi L_{Q_{\rm ra}}(\psi)$;

            Update protagonist policy $\phi \leftarrow \phi-\eta \nabla_\phi L_{\pi_{\rm ra}}(\phi)$;

            Update adversary policy $\beta \leftarrow \beta-\eta \nabla_\beta L_{\mu}(\beta)$;

            Update target network $\bar{\psi} \leftarrow \tau \psi+(1-\tau) \psi$.
        }
    }
\end{algorithm}

\section{Online Safety Verificaiton}

As discussed in previous sections, the RL method for robust reach-avoid problems enjoys considerable scalability and is applicable to complex systems. However, the safety assurance is lost. In this section, we discuss how to reestablish the guarantee of robust constraint satisfaction.

Given an initial state, we can generate a nominal trajectory using the learned reach-avoid policy $\pi_{\rm ra}$.
It is not enough to just check that the reach-avoid value of this nominal trajectory is negative.
We need to further verify that the reach-avoid Values of all possible trajectories that deviate from the nominal one due to disturbances are negative. There are two options to accomplish this goal. The first is to directly check that the constraints are satisfied for all possible realizations of the disturbance sequence, which is similar to the idea of robust open-loop MPC. The second is to permit an additional feedback controller that actively stabilizes the system back to the nominal trajectory and then check the robust constraint satisfaction, which is similar to the idea of tube-based MPC \cite{mayne2014model}. We adopt the second approach, since in the first approach the set of all possible trajectories can grow very large, resulting in misleadingly small safe sets for safety filters.

In tube-based MPC approaches, the stabilizing feedback controllers are typically designed offline and the nominal trajectories are optimized online \cite{mayne2014model}. Recent advances utilize system level synthesis to jointly optimize the nominal trajectory and the feedback controller on-the-fly, with the benefit of significantly reducing conservatism \cite{chen2022robust, leeman2023robust-A, leeman2023robust-B}. However, nonlinear RMPC methods using system level synthesis involve solving non-convex optimization problems, which is computationally prohibitive \cite{leeman2023robust-A, leeman2023robust-B}. Inspired by these works, in this paper, we utilize system level synthesis to only optimize the feedback controller for the smallest worst-case reach-avoid value, and establish a SOCP formulation that can be solved efficiently.

\subsection{Establishing Error Dynamics}

Let $T$ denote the prediction horizon. For a given state $\overline{x}$ and an output $\overline{u}$ from the nominal controller, the nominal trajectory $\mathcal{T}(z_0, T)=\left\{z_0,v_0,z_1,\cdots,z_T,v_T\right\}$ is generated as
\begin{equation}
    \label{generate nominal trajectory}
    \begin{aligned}
        &z_{k+1} = f(z_k, v_k),\quad v_{k+1} = \pi_{\rm ra}(z_{k+1}),\quad k=0,\cdots,T-1,
        \\&\text{s.t.}\quad z_0 = \overline{x},\quad v_0 = \overline{u}.
    \end{aligned}
\end{equation}
in which $\pi_{\rm ra}$ denotes the reach-avoid policy learned with Algorithm \ref{deep RL Policy iteration algorithm}.

To design a feedback controller that stabilizes the system back to the nominal trajectory, firstly we need to establish the error dynamics. For a nonlinear system, it is common to linearize the system around the nominal trajectory and bound the linearization error \cite{koller2018learning, bastani2021safe, leeman2023robust-A}. We adopt the method proposed in \cite{leeman2023robust-A}.

For $0\leq k\leq T$, we have
\begin{equation}
    \begin{aligned}
        f\left(x_k, u_k\right)= & f\left(z_k, v_k\right)+A_k^f\left(x_k-z_k\right)+B_k^f\left(u_k-v_k\right) \\
        & +r_f\left(x_k-z_k, u_k-v_k\right),
    \end{aligned}
\end{equation}
in which $A_k^f\in \mathbb{R}^{n_x \times n_x}$ and $B_k^f\in \mathbb{R}^{n_x \times n_u}$ is given by
\begin{equation}
    A_k^f=\left.\frac{\partial f}{\partial x}\right|_{x=z_k, u=v_k},\quad B_k^f=\left.\frac{\partial f}{\partial u}\right|_{x=z_k, u=v_k}.
\end{equation}
and $r_f\left(x_k-z_k, u_k-v_k\right)$ is the Lagrange remainder. We also have
\begin{equation}
    \begin{aligned}
        g\left(x_k, u_k\right) d_k= & g\left(z_k, v_k\right) d_k+\left(I_{n_x} \otimes d_k^{\top}\right) A_k^{g}\left(x_k-z_k\right)\\&+\left(I_{n_x} \otimes d_k^{\top}\right) B_k^{g}\left(u_k-v_k\right) \\
        & +r_g\left(x_k-z_k, u_k-v_k, d_k\right),
    \end{aligned}
\end{equation}
in which $A_k^g\in \mathbb{R}^{n_x n_d \times n_x}$ and $B_k^g\in \mathbb{R}^{n_x n_d \times n_u}$ is defined as
\begin{equation}
    A_k^g=\left(A_k^{g_{1,:}}, A_k^{g_{2,:}}, \cdots, A_k^{g_{n_x,:}}\right)^{\top},
\end{equation}
\begin{equation}
    B_k^g=\left(B_k^{g_{1,:}}, B_k^{g_{2,:}}, \cdots, B_k^{g_{n_x,:}}\right)^{\top},
\end{equation}
with $A_k^{g_{i,:}}\in \mathbb{R}^{n_x \times n_d}$ and $B_k^{g_{i,:}}\in \mathbb{R}^{n_u \times n_d}$ given by
\begin{equation}
    A_k^{g_{i,:}}=\left.\frac{\partial g_{i,:}}{\partial x}\right|_{x=z_k, u=v_k},\quad B_k^{g_{i,:}}=\left.\frac{\partial g_{i,:}}{\partial u}\right|_{x=z_k, u=v_k}.
\end{equation}
The error term $r_g\left(x_k-z_k, u_k-v_k, d_k\right)$ is determined by $d_k$ and the Lagrange remainder of $g\left(x_k, u_k\right)$.

The overall linearization error $r=r_f\left(x_k-z_k, u_k-v_k\right)+r_g\left(x_k-z_k, u_k-v_k, d_k\right)$ can be overbound by
\begin{equation}
    \label{linearization error bound}
    \begin{aligned}
        &\left|r_i\right| \leq\|e_k\|_{\infty}^2 \mu_i, \quad i=1,2,\cdots, n_x,
        \\& \forall x_k,z_k\in \mathcal{X},\quad \forall u_k,v_k\in \mathcal{U},\quad \forall d_k \in \mathcal{D},
    \end{aligned}
\end{equation}
in which $e_k=\left[(x_k-z_k)^{\top},(u_k-v_k)^{\top}\right]^{\top} \in \mathbb{R}^{n_x+n_u}$ and $\mu_i$ denote the worst-case curvature. See \cite{leeman2023robust-A} for more details.

Define $\Delta x_k = x_k -z_k$ and $\Delta u_k=u_k-v_k$. The error dynamics is given by
\begin{equation}
    \label{error dynamics}
    \begin{aligned}
        &\Delta x_{k+1}=A_k^f \Delta x_k+B_k^f \Delta u_k+w_k,\quad k=0,\cdots,T-1,\\&
        \begin{aligned}
            \text{s.t.}\quad \Delta x_0 =0,\quad \Delta u_0=0,
        \end{aligned}
    \end{aligned}
\end{equation}
in which $w_k$ is given by
\begin{equation}
    \begin{aligned}
        w_k=&r_k+g\left(z_k, v_k\right) d_k+\left(I_{n_x} \otimes d_k^{\top}\right) A_k^g \Delta x_k\\&+\left(I_{n_x} \otimes d_k^{\top}\right) B_k^g \Delta u_k.
    \end{aligned}
\end{equation}

\subsection{SOCP Approach for Worst-case Reach-Avoid Values}

Following the notations in (\ref{LTV system compact form}), we denote the stacked states, nominal states, inputs, nominal inputs and disturbances as $\mathbf{x}=[x_1^{\top},\cdots,x_T^{\top}]^{\top}$, $\mathbf{z}=[z_1^{\top},\cdots,z_T^{\top}]^{\top}$, $\mathbf{u}=[u_1^{\top},\cdots,u_T^{\top}]^{\top}$, $\mathbf{v}=[v_1^{\top},\cdots,v_T^{\top}]^{\top}$, $\mathbf{d}=[d_0^{\top},\cdots,d_{T-1}^{\top}]^{\top}$, respectively. We have $\mathbf{\Delta x} = \mathbf{x}-\mathbf{z}$ and $\mathbf{\Delta u}=\mathbf{u}-\mathbf{v}$. Also let $\mathcal{V} _{\mathcal{D}}$ denote the set of vertices for the polytope $\mathcal{D}$.

Given the error dynamics (\ref{error dynamics}), we can design a feedback controller to stabilize the system back to the nominal trajectory. The controller follows the causal linear state-feedback design, i.e., $\mathbf{\Delta u}=\mathbf{K}\mathbf{\Delta x}$, in which $\mathbf{K} \in \mathcal{L}^{T, n_{u} \times n_{x}}$. The goal is to minimize the worst-case reach-avoid value $V_{\rm ra}^*$ for any closed-loop trajectories of the system (\ref{nonlinear system with disturbance}). If $V_{\rm ra}^*\leq0$, the system is guaranteed to be safely guided back to the target set. Leveraging the tools of system level synthesis, we establish a SOCP formulation to accomplish this goal, which is summarized in (\ref{SOCP formulation}). $n_{H_x}$ denotes the number of rows in $H_x$ and the same goes for $n_{H_u}$ and $n_{R_x}$. The only nonlinear part of (\ref{SOCP formulation}) is (\ref{second-order cone}), which belongs to second-order cone constraints. This is not possible for RMPC approaches even in the linear system case \cite{chen2022robust}, since the nominal trajectory and system responses are both decision variables and are coupled in (\ref{affine space constraints}).

\begin{figure*}
    \begin{subequations}
        \label{SOCP formulation}
        \begin{align}
            \min_{\mathbf{\Phi} _x,\mathbf{\Phi} _u,\mathbf{\Sigma},\lambda ,\eta ,V_{\mathrm{ra}}}&		V_{\mathrm{ra}}\\
            \mathrm{s}.\mathrm{t}.\qquad\,&		\left[ \begin{matrix}
            I-\mathcal{Z} \mathbf{A}&		-\mathcal{Z} \mathbf{B}\\
        \end{matrix} \right] \begin{bmatrix} \mathbf{\Phi}_x \\ \mathbf{\Phi}_u \end{bmatrix} =\mathbf{\Sigma }, \label{affine space constraints}
        \\
            &		\left| I_{n_x}^{i,:}g\left( z_0,v_0 \right) d \right|\le \sigma _{0,i},\quad\forall i=1,\cdots ,n_x,\quad\forall d\in \mathcal{V} _{\mathcal{D}}, \label{trivial disturbance filter overbound}
            \\
            &		\begin{aligned}
            \lambda _k\mu _i&+\left| I_{n_x}^{i,:}g\left( z_k,v_k \right) d \right|+\left\| I_{n_x}^{i,:}\left( I_{n_x}\otimes d^{\top} \right) A_{k}^{g}\mathbf{\Phi} _{x}^{k,1:k} \right\| _1+\left\| I_{n_x}^{i,:}\left( I_{n_x}\otimes d^{\top} \right) B_{k}^{g}\mathbf{\Phi} _{u}^{k,1:k} \right\| _1\le \sigma _{k,i},\\&\forall k=1,\cdots ,T-1,\quad\forall i=1,\cdots ,n_x,\quad\forall d\in \mathcal{V} _{\mathcal{D}}, \label{disturbance filter overbound}
            \\
        \end{aligned}\\
            &		\left\| \begin{pmatrix} \mathbf{\Phi}_x^{k,1:k} \\ \mathbf{\Phi}_u^{k,1:k} \end{pmatrix} \right\| _{\infty}\le \eta _k,\quad\forall k=1,\cdots ,T-1,\label{eta containing}
            \\&
            H_{x}^{i,:}z_0-h_{x}^{i}\le V_{\mathrm{ra}},\quad\forall i=1,\cdots ,n_{H_x},\\
            &		H_{x}^{i,:}z_k+\left\| H_{x}^{i,:}\mathbf{\Phi} _{x}^{k,1:k} \right\| _1-h_{x}^{i}\le V_{\mathrm{ra}},\quad\forall k=1,\cdots ,T,\quad\forall i=1,\cdots ,n_{H_x},\\
            &
            H_{u}^{i,:}v_0-h_{u}^{i}\le V_{\mathrm{ra}},\quad\forall i=1,\cdots ,n_{H_u},\\
            &		H_{u}^{i,:}v_k+\left\| H_{u}^{i,:}\mathbf{\Phi} _{u}^{k,1:k} \right\| _1-h_{u}^{i}\le 0,\quad\forall k=1,\cdots ,T,\quad\forall i=1,\cdots ,n_{H_u},\\
            &		R_{x}^{i,:}z_T+\left\| R_{x}^{i,:}\mathbf{\Phi} _{x}^{T,:} \right\| _1-r_{x}^{i}\le V_{\mathrm{ra}},\quad\forall i=1,\cdots ,n_{R_x},\\
            & \left\|\left(\frac{\lambda_k-1}{2},\eta_k\right)^{\top}\right\|_2\leq \frac{\lambda_k+1}{2},\quad\forall k=1,\cdots ,T-1. \label{second-order cone}
        \end{align}
    \end{subequations}
	{\noindent} \rule[-1pt]{17.85cm}{0.04em}
\end{figure*}

The key of (\ref{SOCP formulation}) is to overbound $w_k$ in (\ref{error dynamics}) with $\Sigma_k \widetilde{w}_k$, in which the disturbance filter $\Sigma_k = {\rm diag}(\sigma_{k,1},\cdots,\sigma_{k,n_x})$ and $\widetilde{w}_k\in\mathcal{B}_{\infty}^{n_x}$. Then utilize Lemma \ref{system level synthesis} to impose state and input constraints on the system responses $\mathbf{\Phi}_x$ and $\mathbf{\Phi}_u$. We formally prove the correctness of this SOCP formulation in the following theorem.

\begin{theorem}[worst-case reach-avoid value]
    \label{worst-case reach-avoid value}
    Let $\mathbf{\Phi}_x^*$, $\mathbf{\Phi}_u^*$ and $V^*_{\rm ra}$ denote the optimal solution of (\ref{SOCP formulation}). The overall feedback controller is constructed as $\mathbf{u} = \mathbf{v} + \mathbf{K}(\mathbf{x}-\mathbf{z})$, in which $\mathbf{K}=\mathbf{\Phi}_u^* {\mathbf{\Phi}_x^*}^{-1}$. Any close-loop trajecotry $\mathcal{T}(x_0,T)$ of the system (\ref{nonlinear system with disturbance}) under this feedback controller with $\mathbf{d}\in \mathcal{D}^T$ satisfies that $V_{\rm ra}(\mathcal{T}(x_0,T))\leq V^*_{\rm ra}$.
\end{theorem}

\begin{proof}
    We first prove that $w_k$ in (\ref{error dynamics}) is overbounded by $\Sigma_k \widetilde{w}_k$, in which the disturbance filter $\Sigma_k = {\rm diag}(\sigma_{k,1},\cdots,\sigma_{k,n_x})$ and $\widetilde{w}_k\in\mathcal{B}_{\infty}^{n_x}$. With this overbound, we can conclude all possible behaviors of the error dynamics (\ref{error dynamics}) are included in those of the new error dynamics as $\Delta x_{k+1}=A_k^f \Delta x_k+B_k^f \Delta u_k+\Sigma_k \widetilde{w}_k$. Then we can utilize Lemma \ref{system level synthesis} to obtain that $\mathbf{\Delta x} = \mathbf{\Phi}_x \mathbf{\widetilde{w}}$ and $\mathbf{\Delta u} = \mathbf{\Phi}_u \mathbf{\widetilde{w}}$ for $\mathbf{\widetilde{w}}\in\mathcal{B}_{\infty}^{Tn_x}$. It is suffice is to show that $|I_{n_x}^{i,:}w_k|\leq\sigma_{k,i}$.
    For $k=0$, it is obvious in (\ref{trivial disturbance filter overbound}).
    For $1\leq k\leq T-1$, we have
    \begin{equation}
        \begin{aligned}
            |I_{n_x}^{i,:}w_k|\leq& |I_{n_x}^{i,:}r_k|+|I_{n_x}^{i,:}g\left(z_k, v_k\right) d_k|\\&+| I_{n_x}^{i,:}\left( I_{n_x}\otimes d_k^{\top} \right) A_{k}^{g}\Delta x_k |\\&+| I_{n_x}^{i,:}\left( I_{n_x}\otimes d_k^{\top} \right) B_{k}^{g}\Delta u_k |.
        \end{aligned}
    \end{equation}
    Utilizing (\ref{linearization error bound}), we have $|I_{n_x}^{i,:}r_k|\leq \left\|(\Delta x_k,\Delta u_k)^{\top}\right\|_{\infty}\mu_i$.
    With Lemma \ref{system level synthesis}, we have $\Delta x_k = \mathbf{\Phi}_x^{k,1:k}\mathbf{\widetilde{w}}^{1:k n_x}$ and $\Delta u_k = \mathbf{\Phi}_u^{k,1:k}\mathbf{\widetilde{w}}^{1:k n_x}$. Since $\left\|\mathbf{\widetilde{w}}^{1:k n_x}\right\|_\infty=1$, together with (\ref{eta containing}) and (\ref{second-order cone}), we have $|I_{n_x}^{i,:}r_k|\leq \lambda_k\mu_i$. With the definition of matrix one-norm, we have $| I_{n_x}^{i,:}\left( I_{n_x}\otimes d_k^{\top} \right) A_{k}^{g}\Delta x_k |\leq \left\| I_{n_x}^{i,:}\left( I_{n_x}\otimes d^{\top} \right) A_{k}^{g}\mathbf{\Phi} _{x}^{k,1:k} \right\| _1$ and $| I_{n_x}^{i,:}\left( I_{n_x}\otimes d_k^{\top} \right) B_{k}^{g}\Delta x_k |\leq \left\| I_{n_x}^{i,:}\left( I_{n_x}\otimes d^{\top} \right) B_{k}^{g}\mathbf{\Phi} _{u}^{k,1:k} \right\| _1$. Putting the above analysis together, we have
    \begin{equation}
        \label{last step for overbound}
        \begin{aligned}
            |I_{n_x}^{i,:}w_k|\leq& |I_{n_x}^{i,:}r_k|+|I_{n_x}^{i,:}g\left(z_k, v_k\right) d_k|\\&+\left\| I_{n_x}^{i,:}\left( I_{n_x}\otimes d_k^{\top} \right) A_{k}^{g}\mathbf{\Phi} _{x}^{k,1:k} \right\| _1\\&+\left\| I_{n_x}^{i,:}\left( I_{n_x}\otimes d_k^{\top} \right) B_{k}^{g}\mathbf{\Phi} _{u}^{k,1:k} \right\| _1,
        \end{aligned}
    \end{equation}
    which aligns with (\ref{disturbance filter overbound}).
    Since the right-hand side of (\ref{last step for overbound}) is linear with $d$ and $\mathcal{D}$ is convex, we only need (\ref{disturbance filter overbound}) to hold for all $d\in \mathcal{V} _{\mathcal{D}}$.

    Next we show that the worst-case reach-avoid value is bounded by $V_{\rm ra}$. For state constraints, we have
    \begin{equation}
        \begin{aligned}
            H_x^{i,:}x_k-h_x^i&= H_x^{i,:}(z_k+\Delta x_k)-h_x^i\\&
            = H_x^{i,:}z_k+H_x^{i,:}\Delta x_k-h_x^i\\&
            \leq H_x^{i,:}z_k+\left\| H_x^{i,:}\mathbf{\Phi} _{x}^{k,1:k} \right\| _1-h_x^i,
        \end{aligned}
    \end{equation}
    in which the last inequality follows from the definition of vector one-norm and $\Delta x_k = \mathbf{\Phi}_x^{k,1:k}\mathbf{\widetilde{w}}^{1:k n_x}$. Similarly, the bounds for input constraints and target constraints can also be proved. Also, with Lemma (\ref{system level synthesis}), the overall feedback controller is given by $\mathbf{u} = \mathbf{v} + \mathbf{K}^*(\mathbf{x}-\mathbf{z})$.
\end{proof}

The execution procedure of the proposed safety filter is summarized in Algorithm \ref{Safety filter design}. We show in the following theorem that this algorithm achieves persistent safety guarantee.

\begin{theorem}[persistent safety guarantee]
    \label{persistent safety guarantee}
    Assume that the first-time computation of (\ref{SOCP formulation}) yields a $V_{\rm ra}^*\leq0$. Then the infinite-horizon closed-loop trajectory of the system (\ref{nonlinear system with disturbance}) under the safety filter described with Algorithm \ref{Safety filter design} is guaranteed to satisfy the state constraints, i.e., $x_k\in\mathcal{X}$, $\forall k\geq0$ and $\forall d_k\in\mathcal{D}$.
\end{theorem}

\begin{proof}
    Given a $V_{\rm ra}^*\leq0$ and the corresponding $\mathbf{\Phi}_x^*$ and $\mathbf{\Phi}_u^*$, a closed-loop feedback controller is constructed as $\mathbf{u} = \mathbf{v} + \mathbf{K}(\mathbf{x}-\mathbf{z})$, in which $\mathbf{K}=\mathbf{\Phi}_u^* {\mathbf{\Phi}_x^*}^{-1}$. Based on Theorem \ref{worst-case reach-avoid value}, the closed-loop trajectory under this feedback controller satisfies that $x_k\in\mathcal{X}$, $0\leq k\leq T$ and $x_T\in \mathcal{R}$. As shown in Algorithm \ref{Safety filter design}, at every following state, the SOCP problem (\ref{SOCP formulation}) is solved. There may be two cases. If a $V_{\rm ra}^*>0$ is obtained, we choose the initially constructed feedback controller $\mathbf{u} = \mathbf{v} + \mathbf{K}(\mathbf{x}-\mathbf{z})$ for $1\leq k\leq T$ and the terminal controller $\pi_{\rm terminal}$ for $k\geq T+1$. Therefore, the state trajectory satisfies $x_k\in\mathcal{X}$. If a $V_{\rm ra}^*\leq0$ is obtained, we reset $k=0$ and reconstruct the closed-loop feedback controller, which gets the situation back to the beginning. In either case, the safety of the system is guaranteed.
\end{proof}

\section{Numerical Example}\label{Numerical Example}

In this section, we demonstrate the effectiveness of our method through a numerical example. Consider the dynamics of a pendulum:
\begin{equation}
    \begin{bmatrix}
        \dot{x}_2 \\
        \dot{x}_2
        \end{bmatrix}=\begin{bmatrix}
        x_1 \\
        \frac{3g}{2l}\sin \left(x_1\right)
        \end{bmatrix}+\begin{bmatrix}
            0 \\
            \frac{3}{ml^2} u
            \end{bmatrix}+\begin{bmatrix}
                d_1 \\
                d_2 + d_3 u
                \end{bmatrix},
\end{equation}
in which $x_1$ and $x_2$ denote the angle and angular velocity of the pendulum, respectively. The control input $u$ is the torque applied to the pendulum. The parameters are $g=10.0$, $l=1.0$, $m=1.0$. The input constraint is $\mathcal{U}=\left\{u\mid -5\leq u\leq 5\right\}$. The disturbances are $d_1\in[-0.01,0.01]$, $d_2\in[-0.01,0.01]$ and $d_3\in[-0.001,0.001]$. The state constraint is $\mathcal{X}=\left\{(x_1,x_2)\mid x_1\in [-\pi/3,\pi/3],x_2\in[-2.0,2.0]\right\}$. The system is discretized with the fourth-order Runge-Kutta method and the sampling time is $0.05$s. The target set is defined as $\mathcal{R}=\left\{(x_1,x_2)\mid x_1\in [-\pi/12,\pi/12],x_2\in[-0.5,0.5]\right\}$, i.e., a small box around the origin and can be rendered safe with a linear quadratic regulator (LQR) as the terminal safe policy.

We train a neural network $\pi_{\rm ra}$ with Algorithm \ref{deep RL Policy iteration algorithm}. As discussed in previous sections, the most important evaluation metric for safety filters is the size of safe sets determined by them. We calculate the maximal RIS by solving the HJI PDE with the level set toolbox \cite{mitchell2002application}. To evaluate the size of the safe set rendered by the proposed safety filter, we employ a $40\times60$ grid in $\mathcal{X}$ and utilize the grid points as initial states $\overline{x}$. The nominal input is chosen as $\overline{u}=\pi_{\rm ra}(\overline{x})$. We solve the SOCP problem (\ref{SOCP formulation}) at these initial states and collect the reach-avoid value $V_{\rm ra}^*$. We also record the time consumed by solving the SOCP problem. The safe set of the safety filter contains those initial states with $V_{\rm ra}^*\leq 0$. The SOCP is formulated with CasADi \cite{andersson2019casadi} and solved with Gurobi. The prediction horizon $T$ is 25.

We compare our proposed safety filter with a relevant method presented by Bastani et al. \cite{bastani2021safe}. They propose to train a neural network recovery policy $\pi_{\rm rec}$ with standard RL methods and use it for model predictive shielding. The recovery policy is trained to reach the target set as fast as possible, i.e., the agent gets punished when not reaching the target set. We train a neural network $\pi_{\rm rec}$ with soft actor-critic \cite{haarnoja2018soft}, a state-of-the-art RL algorithm, following the description in \cite{bastani2021safe}.
Note that $\pi_{\rm ra}$ and $\pi_{\rm rec}$ are trained with the same amount of data and the shared parameters (e.g., network architecture, actor learning rate) are also the same.
For safety verification, we do not utilize the sampling-based method proposed in \cite{bastani2021safe} since it only provides probabilistic guarantee. We use our SOCP formulation (\ref{SOCP formulation}) to evaluate the safe set of $\pi_{\rm rec}$.

We also compare our method with the RMPC approach proposed in \cite{leeman2023robust-A, leeman2023robust-B}, which utilizes system level synthesis to jointly optimize the nominal trajectory and the stabilizing linear feedback controller. The objective is chosen as $J=v_0^2$, i.e., minimize the first control input. Since the optimization problem formulated by RMPC is non-convex, we resort to the Ipopt solver. We employ the $40\times60$ grid in $\mathcal{X}$ and utilize the grid points as initial states for RMPC and record the time consumed by solving the optimization problem (either returns a solution or reports infeasibility). Since the safety filtering problem is more like checking the feasibility of RMPC rather than solving it, we employ an optimization formulation that encodes the RMPC feasibility checking for additional comparison. Suppose an optimization problem is specified as
\begin{equation}
    \label{initial optimization problem}
    \begin{aligned}
        & \min _x f(x) \\
        & \text { s.t. }\left\{\begin{array}{l}
        c^{\mathcal{E}}(x)=0 \\
        c^{\mathcal{I}}(x) \leq 0
        \end{array}\right..
        \end{aligned}
\end{equation}
The feasibility checking for (\ref{initial optimization problem}) is formulated as
\begin{equation}
    \label{feasibility checking}
    \begin{aligned}
        & \min _{x,t} \quad t^{\top}t \\
        & \text { s.t. }\left\{\begin{array}{l}
        -t_i^{\mathcal{E}}\leq c_i^{\mathcal{E}}(x)\leq t_i^{\mathcal{E}} \\
        c_j^{\mathcal{I}}(x)\leq t_j^{\mathcal{I}}\\
        t_i^{\mathcal{E}}\geq 0\\
        t_j^{\mathcal{I}}\geq 0
        \end{array}\right.,
    \end{aligned}
\end{equation}
in which $t=(t_{\mathcal{E}}^{\top},t_{\mathcal{I}}^{\top})^{\top}$ denotes the constraint violation. The optimization problem (\ref{initial optimization problem}) is treated as feasible if the optimal solution of (\ref{feasibility checking}) is below $10^{-4}$.

The evaluation results of safe sets are shown in Fig. \ref{safe set results}.
The maximal RIS is calculated by solving the HJI PDE. We plot the zero-value contour of the solved safety value, which is denoted by the green line.
For $\pi_{\rm ra}$ and $\pi_{\rm rec}$, we draw the zero-value contour of their solved $V_{\rm ra}*$, which is denoted by blue line and red line, respectively. The results demonstrate the effectiveness of the reach-avoid method, since the safe set rendered by $\pi_{\rm ra}$ is much larger than that of $\pi_{\rm rec}$. $\pi_{\rm rec}$ is trained with the goal of reaching the target set as fast as possible, which is highly suboptimal for safety filtering. Moreover, the safe set of $\pi_{\rm ra}$ is very close to the maximal RIS, with marginal differences at the top right and the bottom left.

\begin{figure}[htbp]
    \centering
    \includegraphics[width=3.5in]{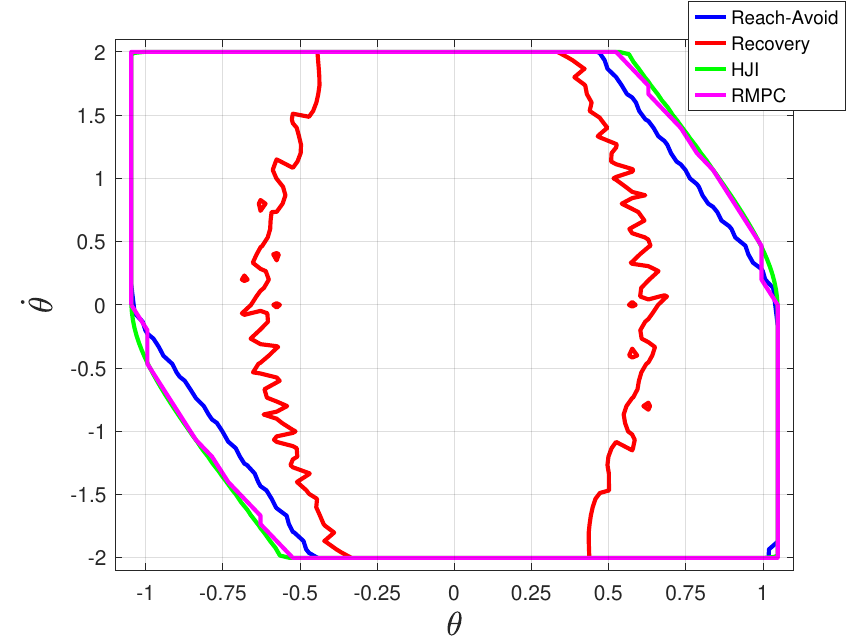}
    \caption{The safe sets for different methods.}
    \label{safe set results}
\end{figure}

The boundary of the feasible states for RMPC \cite{leeman2023robust-B} is denoted by the magenta line. It almost coincides with the green line (i.e., the boundary of the maximal RIS). This demonstrates the effectiveness of jointly optimizing nominal trajectories and stabilizing controllers on-the-fly, which significantly reduces conservativeness. However, the RMPC approach is much more computationally expensive compared to our proposed safety filter, as shown in Table \ref{Computation Time Comparison} (SD denotes standard deviation). Note that the computation time for the three methods in Table \ref{Computation Time Comparison} are evaluated on the same computer. The mean computation time of RMPC or its feasibility checking formulation (\ref{feasibility checking}) is about 40-50 times larger than the safety filter. It is very challenging for RMPC to realize infeasibility when the system is outside the maximal RIS, so the maximal computation time for RMPC is very large. The feasibility checking formulation performs better, but is still much slower than the safety filter.

\begin{table}[htbp]
    \caption{Computation Time Comparison}
    \label{Computation Time Comparison}
    \centering
    \begin{tabular}{ l c c c }
    Method & Mean & SD & Max\\ \hline
    Proposed safety filter & 0.1446 s & 0.0225 s &  0.2558 s\\
    Nonlinear RMPC from \cite{leeman2023robust-B} & 5.1962 s & 7.9486 s & 42.1694 s\\ 
    Feasibility checking for RMPC & 4.3684 s & 0.5716 s & 6.7651 s\\\hline 
    \end{tabular}
\end{table}

\section{Conclusion}

In this work, we present a theoretical framework that bridges the advantages of both RMPC and RL to synthesize safety filters for nonlinear systems with state- and action-dependent uncertainty. We decompose the (maximal) RIS into two parts: a target set that aligns with terminal region design of RMPC, and a reach-avoid set that accounts for the rest of RIS. A policy iteration approach is proposed for robust reach-avoid problems and its monotone convergence is proved. This method paves the way for an actor-critic deep RL algorithm, which simultaneously synthesizes a reach-avoid policy network, a disturbance policy network, and a reach-avoid value network. The learned reach-avoid policy network is utilized to generate nominal trajectories for online verification. We formulate a SOCP approach for online verification using system level synthesis, which optimizes for the worst-case reach-avoid value of any possible trajectories. The proposed safety filter requires much lower computational complexity than RMPC and still enjoys persistent robust safety guarantee. The effectiveness of our method is illustrated through a numerical example.

\section*{Acknowledgement}

The authors would like to thank Haitong Ma and Yunan Wang for valuable suggestions on problem formulation. The authors also would like to thank Antoine P. Leeman for helpful discussions on system level synthesis.

\bibliographystyle{IEEEtran}
\bibliography{reference}

\begin{IEEEbiography}[{\includegraphics[width=1in,height=1.25in,clip,keepaspectratio]{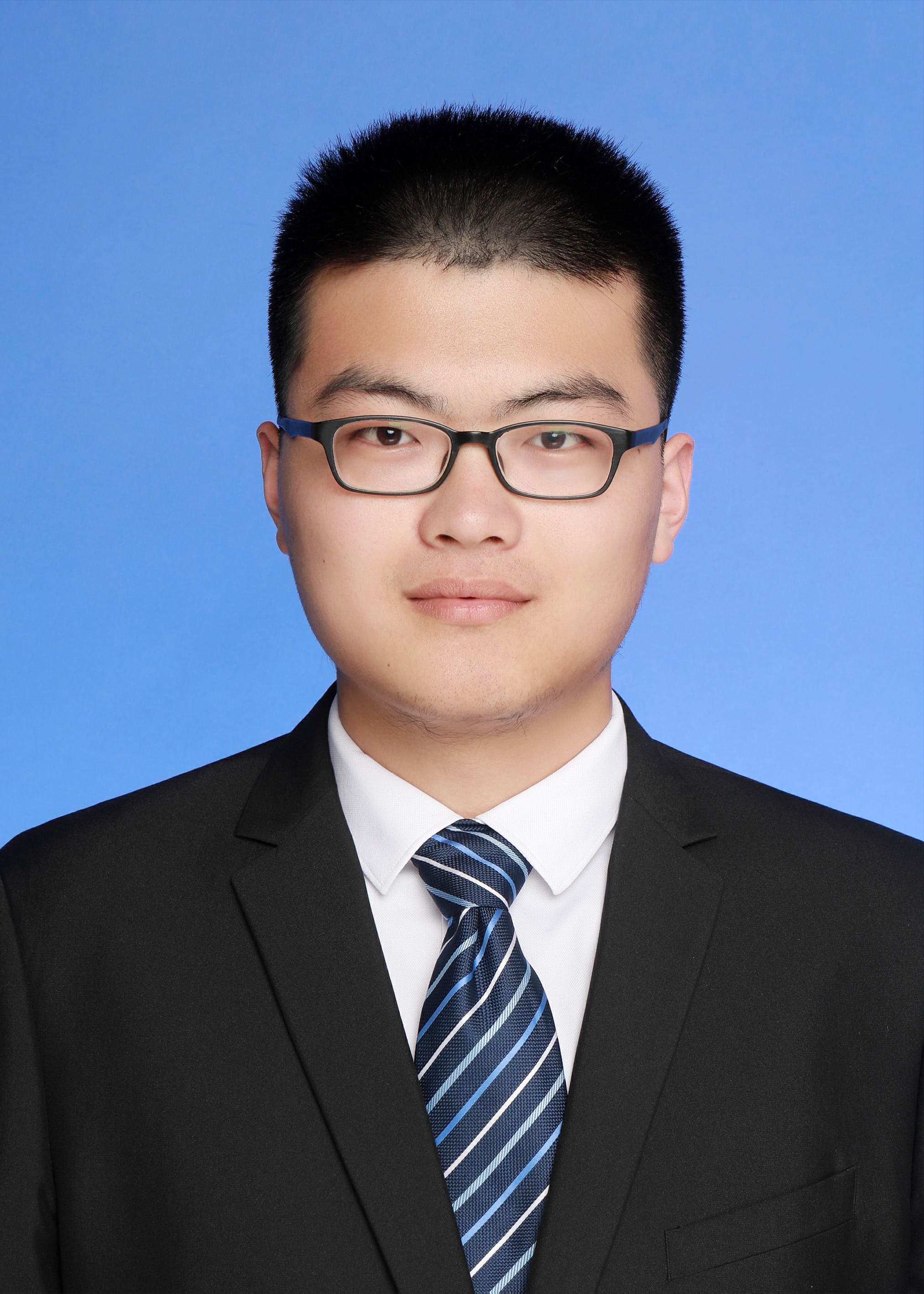}}]{Zeyang Li}received the B.S. degree in mechanical engineering in 2021, from the School of Mechanical Engineering, Shanghai Jiao Tong University, Shanghai, China. He is currently working toward the M.S. degree in mechanical engineering with the Department of Mechanical Engineering, Tsinghua University, Beijing, China. His research interests include reinforcement learning and optimal control.
\end{IEEEbiography}

\vspace{50 mm}

\begin{IEEEbiography}[{\includegraphics[width=1in,height=1.25in,clip,keepaspectratio]{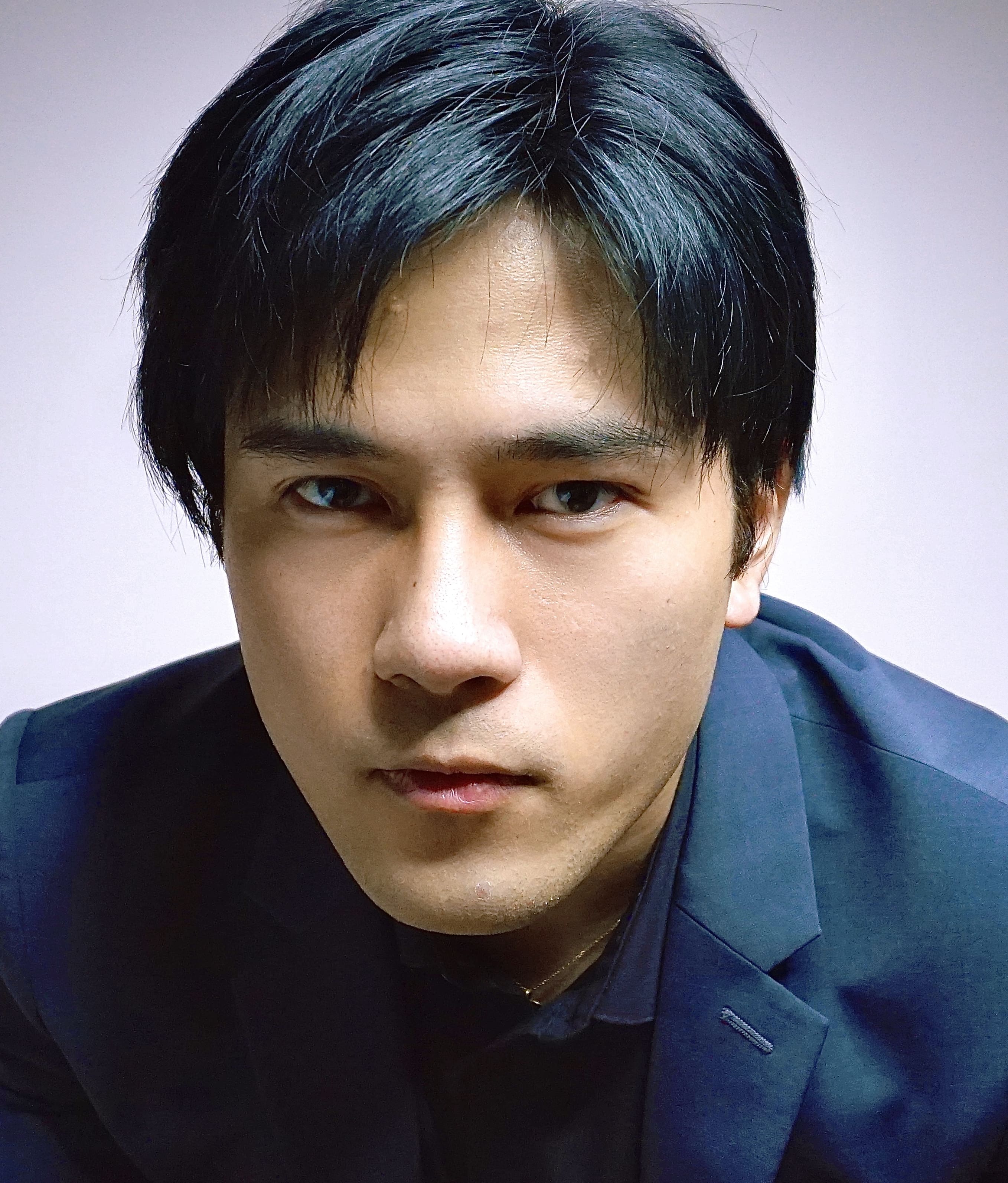}}]{Weiye Zhao}is a fifth-year PhD candidate at the Robotics Institute of Carnegie Mellon University, with a primary research focus on providing provable safety guarantees for robots. Specifically, he is interested in safe reinforcement learning, nonconvex optimization, and safe control. He has published in numerous renowned venues, including IJCAI, ACC, CDC, AAAI, CoRL, L4DC, RCIM, LCSS, and others. He serves on the program committee for several conferences, including IV, AAAI, IJCAI, CoRL, CDC, ACC, L4DC, NeurIPS and more.
\end{IEEEbiography}

\begin{IEEEbiography}[{\includegraphics[width=1in,height=1.25in,clip,keepaspectratio]{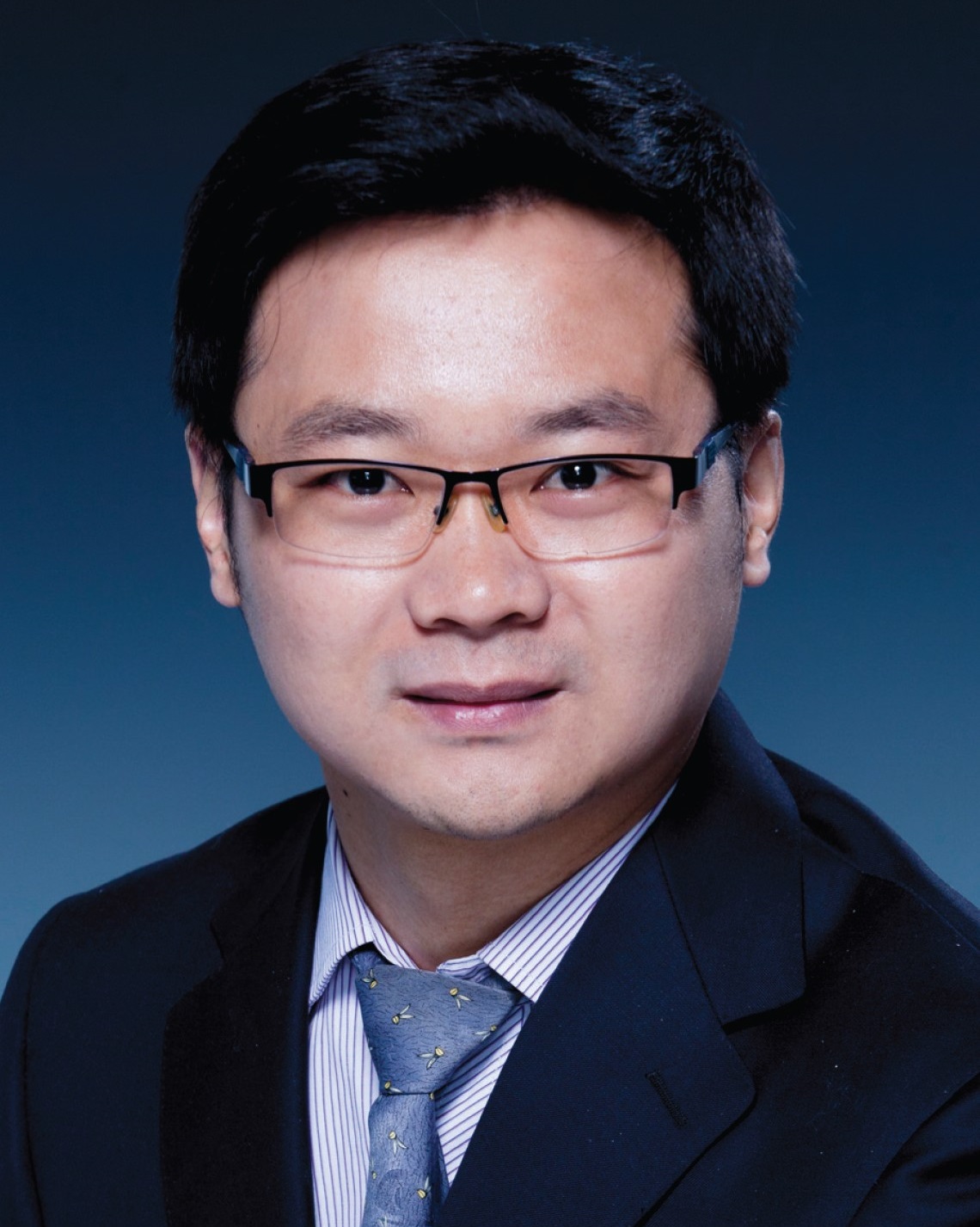}}]{Chuxiong Hu}(Senior Member, IEEE) received his B.S. and Ph.D. degrees in Mechatronic Control Engineering from Zhejiang University, Hangzhou, China, in 2005 and 2010, respectively. He is currently an Associate Professor (tenured) at Department of Mechanical Engineering, Tsinghua University, Beijing, China. From 2007 to 2008, he was a Visiting Scholar in mechanical engineering with Purdue University, West Lafayette, USA. In 2018, he was a Visiting Scholar in mechanical engineering with University of California, Berkeley, CA, USA. His research interests include precision motion control, high-performance multiaxis contouring control, precision mechatronic systems, intelligent learning, adaptive robust control, neural networks, iterative learning control, and robot. Prof. Hu was the recipient of the Best Student Paper Finalist at the 2011 American Control Conference, the 2012 Best Mechatronics Paper Award from the ASME Dynamic Systems and Control Division, the 2013 National 100 Excellent Doctoral Dissertations Nomination Award of China, the 2016 Best Paper in Automation Award, the 2018 Best Paper in AI Award from the IEEE International Conference on Information and Automation, and 2022 Best Paper in Theory from the IEEE/ASME International Conference on Mechatronic, Embedded Systems and Applications. He is now an Associate Editor for the IEEE Transactions on Industrial Informatics and a Technical Editor for the IEEE/ASME Transactions on Mechatronics.
\end{IEEEbiography}

\begin{IEEEbiography}[{\includegraphics[width=1in,height=1.25in,clip,keepaspectratio]{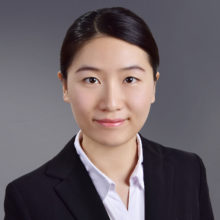}}]{Changliu Liu} is an assistant professor in the Robotics Institute, Carnegie Mellon University (CMU), where she leads the Intelligent Control Lab. Prior to joining CMU, Dr. Liu was a postdoc at Stanford Intelligent Systems Laboratory. She received her Ph.D. from University of California at Berkeley and her bachelor degrees from Tsinghua University. Her research interests lie in human-robot collaborations. Her work is recognized by NSF Career Award, Amazon Research Award, and Ford URP Award.
\end{IEEEbiography}

\vfill

\end{document}